\documentclass{article}
\usepackage[utf8]{inputenc}
\usepackage{amsthm}
\usepackage{amsfonts,mathrsfs}
\usepackage{authblk}
\usepackage{enumerate}
\usepackage{algorithm2e}
\def\showauthornotes{0}

\usepackage{etex}

\usepackage[l2tabu, orthodox]{nag}

\usepackage{xspace,enumerate}

\usepackage[dvipsnames]{xcolor}

\usepackage[T1]{fontenc}
\usepackage[full]{textcomp}

\usepackage[american]{babel}

\usepackage{mathtools}

\usepackage{stmaryrd}

\usepackage{amsthm}

\newtheorem{theorem}{Theorem}[section]
\newtheorem*{theorem*}{Theorem}

\newtheorem{proposition}[theorem]{Proposition}
\newtheorem*{proposition*}{Proposition}
\newtheorem{lemma}[theorem]{Lemma}
\newtheorem*{lemma*}{Lemma}

\newtheorem*{conjecture*}{Conjecture}
\newtheorem{fact}[theorem]{Fact}
\newtheorem*{fact*}{Fact}

\newtheorem*{hypothesis*}{Hypothesis}

\theoremstyle{definition}
\newtheorem{definition}[theorem]{Definition}

\theoremstyle{remark}
\newtheorem{claim}[theorem]{Claim}
\newtheorem*{claim*}{Claim}
\newtheorem{remark}[theorem]{Remark}
\newtheorem*{remark*}{Remark}

\newtheorem*{observation*}{Observation}

\newtheorem*{rep@theorem}{\rep@title}
\newcommand{\newreptheorem}[2]{%
\newenvironment{rep#1}[1]{%
 \def\rep@title{#2 ##1}%
 \begin{rep@theorem}}%
 {\end{rep@theorem}}
}

\usepackage[letterpaper,
top=1.3in,
bottom=1.3in,
left=1.6in,
right=1.6in]{geometry}

\usepackage[varg]{pxfonts} 

\usepackage{tgpagella}
\usepackage{lmodern}

\def\showcolorlinks{1}
\ifnum\showcolorlinks=1
\usepackage[
pagebackref,
colorlinks=true,
urlcolor=blue,
linkcolor=blue,
citecolor=OliveGreen,
]{hyperref}
\fi

\ifnum\showcolorlinks=0
\usepackage[
pagebackref,
letterpaper=true,
colorlinks=false,
pdfborder={0 0 0}
]{hyperref}
\fi

\usepackage{prettyref}

\newcommand{\savehyperref}[2]{\texorpdfstring{\hyperref[#1]{#2}}{#2}}

\newrefformat{eq}{\savehyperref{#1}{\textup{(\ref*{#1})}}}
\newrefformat{lem}{\savehyperref{#1}{Lemma~\ref*{#1}}}
\newrefformat{def}{\savehyperref{#1}{Definition~\ref*{#1}}}
\newrefformat{thm}{\savehyperref{#1}{Theorem~\ref*{#1}}}
\newrefformat{cor}{\savehyperref{#1}{Corollary~\ref*{#1}}}
\newrefformat{cha}{\savehyperref{#1}{Chapter~\ref*{#1}}}
\newrefformat{sec}{\savehyperref{#1}{Section~\ref*{#1}}}
\newrefformat{app}{\savehyperref{#1}{Appendix~\ref*{#1}}}
\newrefformat{tab}{\savehyperref{#1}{Table~\ref*{#1}}}
\newrefformat{fig}{\savehyperref{#1}{Figure~\ref*{#1}}}
\newrefformat{hyp}{\savehyperref{#1}{Hypothesis~\ref*{#1}}}
\newrefformat{rem}{\savehyperref{#1}{Remark~\ref*{#1}}}
\newrefformat{item}{\savehyperref{#1}{Item~\ref*{#1}}}
\newrefformat{step}{\savehyperref{#1}{step~\ref*{#1}}}
\newrefformat{conj}{\savehyperref{#1}{Conjecture~\ref*{#1}}}
\newrefformat{fact}{\savehyperref{#1}{Fact~\ref*{#1}}}
\newrefformat{prop}{\savehyperref{#1}{Proposition~\ref*{#1}}}
\newrefformat{prob}{\savehyperref{#1}{Problem~\ref*{#1}}}
\newrefformat{claim}{\savehyperref{#1}{Claim~\ref*{#1}}}
\newrefformat{relax}{\savehyperref{#1}{Relaxation~\ref*{#1}}}
\newrefformat{red}{\savehyperref{#1}{Reduction~\ref*{#1}}}
\newrefformat{part}{\savehyperref{#1}{Part~\ref*{#1}}}

\newcommand{\Sref}[1]{\hyperref[#1]{\S\ref*{#1}}}

\usepackage{nicefrac}

\def\usemicrotype{0}
\ifnum\usemicrotype=1
\usepackage{microtype}
\fi

\ifnum\showauthornotes=1
\newcommand{\Authornote}[2]{{\sffamily\small\color{red}{[#1: #2]}}}
\newcommand{\Authornotecolored}[3]{{\small\color{#1}{[#2: #3]}}}
\newcommand{\Authorcomment}[2]{{\sffamily\small\color{gray}{[#1: #2]}}}
\newcommand{\Authorstartcomment}[1]{\sffamily\small\color{gray}[#1: }

\newcommand{\Authorfnote}[2]{\footnote{\color{red}{#1: #2}}}
\newcommand{\Authorfixme}[1]{\Authornote{#1}{\textbf{??}}}
\newcommand{\Authormarginmark}[1]{\marginpar{\textcolor{red}{\fbox{\Large #1:!}}}}
\else
\newcommand{\Authornote}[2]{}
\newcommand{\Authornotecolored}[3]{}
\newcommand{\Authorcomment}[2]{}
\newcommand{\Authorstartcomment}[1]{}

\newcommand{\Authorfnote}[2]{}
\newcommand{\Authorfixme}[1]{}
\newcommand{\Authormarginmark}[1]{}
\fi

\def\showfixme{0}
\ifnum\showfixme=1

\fi

\usepackage{boxedminipage}

\newcommand{\card}[1]{\lvert#1\rvert}

\newcommand\sett[2]{\left\{ #1 \left| \; \vphantom{#1 #2} \right. #2  \right\}}
\newcommand{\set}[1]{\{#1\}}

\newcommand{\norm}[1]{\lVert#1\rVert}

\newcommand{\iprod}[1]{\langle#1\rangle}

\def\dim{{\textrm{dim}}}

\newcommand{\Esymb}{\mathbb{E}}
\newcommand{\Psymb}{\mathbb{P}}

\DeclareMathOperator*{\E}{\Esymb}

\DeclareMathOperator*{\ProbOp}{\Psymb}

\renewcommand{\Pr}{\ProbOp}
\def\one{{\mathbf{1}}}
\def\ind{{\mathds{1}}}

\newcommand{\ve}{\;\hbox{and}\;}

 \usepackage{dsfont}
\usepackage{mathrsfs}

\newcommand{\textparen}[1]{\text{(#1)}}

\ifx\because\undefined
\newcommand{\because}[1]{\textparen{because #1}}
\else
\renewcommand{\because}[1]{\textparen{because #1}}
\fi

\newcommand{\bits}{\mathbb{F}_2}

\newcommand\bdot\bullet

\DeclareMathOperator{\poly}{poly}

\DeclareMathOperator{\dist}{dist}

\newcommand{\N}{\mathbb N}
\newcommand{\R}{\mathbb R}

\renewcommand{\leq}{\leqslant}

\renewcommand{\geq}{\geqslant}
\renewcommand{\ge}{\geqslant}

\let\epsilon=\varepsilon

\numberwithin{equation}{section}

\newcommand{\MYstore}[2]{%
  \global\expandafter \def \csname MYMEMORY #1 \endcsname{#2}%
}

\newcommand{\MYload}[1]{%
  \csname MYMEMORY #1 \endcsname%
}

\newcommand{\MYnewlabel}[1]{%
  \newcommand\MYcurrentlabel{#1}%
  \MYoldlabel{#1}%
}

\newcommand{\MYdummylabel}[1]{}

\newcommand{\torestate}[1]{%
  \let\MYoldlabel\label%
  \let\label\MYnewlabel%
  #1%
  \MYstore{\MYcurrentlabel}{#1}%
  \let\label\MYoldlabel%
}

\newcommand{\restatetheorem}[1]{%
  \let\MYoldlabel\label
  \let\label\MYdummylabel
  \begin{theorem*}[Restatement of \prettyref{#1}]
    \MYload{#1}
  \end{theorem*}
  \let\label\MYoldlabel
}

\newcommand{\restatelemma}[1]{%
  \let\MYoldlabel\label
  \let\label\MYdummylabel
  \begin{lemma*}[Restatement of \prettyref{#1}]
    \MYload{#1}
  \end{lemma*}
  \let\label\MYoldlabel
}

\newcommand{\restateprop}[1]{%
  \let\MYoldlabel\label
  \let\label\MYdummylabel
  \begin{proposition*}[Restatement of \prettyref{#1}]
    \MYload{#1}
  \end{proposition*}
  \let\label\MYoldlabel
}

\newcommand{\restatefact}[1]{%
  \let\MYoldlabel\label
  \let\label\MYdummylabel
  \begin{fact*}[Restatement of \prettyref{#1}]
    \MYload{#1}
  \end{fact*}
  \let\label\MYoldlabel
}

\newcommand{\restate}[1]{%
  \let\MYoldlabel\label
  \let\label\MYdummylabel
  \MYload{#1}
  \let\label\MYoldlabel
}

\newcommand{\eps}{\epsilon}

\let\origparagraph\paragraph
\renewcommand{\paragraph}[1]{\origparagraph{#1.}}

\allowdisplaybreaks

\def\F{\mathbb{F}}

\def\mpar{M^{||}}
\def\gpar{E^{||}}
\def\ga{E^{A}}
\def\gb{E^{B}}
\def\bc{{\textsf{C}}}
\def\disagr{{\Delta}}
\newcommand\remove[1]{}
\newcommand\T[1]{{X}_{#1}}
\def\L{\mathrm{w}}
\def\W{{W}}
\def\dcol{\delta^{\textrm{col}}}
\def\drow{\delta^{\textrm{row}}}
\def\D{{\mathcal{D}}}
\def\Up{{\mathsf U}}
\def\Down{{\mathsf D}}
\def\vrej{\zeta}
\def\Rate{\mbox{Rate}}

\remove{
\title{Locally Testable Codes\\ with constant rate, distance, and locality}
\author[1]{Irit Dinur\thanks{I.D. acknowledges support by ERC grant 772839 and ISF grant 2073/21.}}
\author[2]{Shai Evra\thanks{S.E. is grateful to the Azrieli Foundation for the award of an Azrieli Fellowship.}}
\author[2]{Ron Livne}
\author[1]{Alexander Lubotzky\thanks{A.L. acknowledges support by ERC grant 882751 and a grant from the Institute for Advanced Study at Princeton.}}
\author[2]{Shahar Mozes\thanks{S.M. acknowledges support by ISF-Moked grant 2019/19.}}
\affil[1]{Weizmann Institute, Rehovot, Israel}
\affil[2]{Hebrew University, Jerusalem, Israel}
}
\title{Locally Testable Codes\\ with constant rate, distance, and locality}
\author[1]{Irit Dinur}
\author[2]{Shai Evra}
\author[2]{Ron Livne}
\author[1]{Alexander Lubotzky}
\author[2]{Shahar Mozes}
\affil[1]{Weizmann Institute, Rehovot, Israel}
\affil[2]{Hebrew University, Jerusalem, Israel}

\begin{document}

\maketitle

\begin{abstract}    
A locally testable code (LTC) is an error correcting code that has a property-tester. The tester reads $q$ bits that are randomly chosen, and rejects words with probability proportional to their distance from the code. The parameter $q$ is called the locality of the tester.

LTCs were initially studied as important components of PCPs, and since then the topic has evolved on its own. High rate LTCs could be useful in practice: before attempting to decode a received word, one can save time by first quickly testing if it is close to the code.

An outstanding open question has been whether there exist ``$c^3$-LTCs'', namely LTCs with \textbf{c}onstant rate, \textbf{c}onstant distance, and \textbf{c}onstant locality.

In this work we construct such codes based on a new two-dimensional complex which we call a left-right Cayley complex. This is essentially a graph which, in addition to  vertices and edges, also has squares. Our codes can be viewed as a two-dimensional version of (the one-dimensional) expander codes, where the codewords are functions on the squares rather than on the edges. 
\end{abstract}

\section{Introduction}

A locally testable code (LTC) is an error correcting code that has a property-tester. The tester reads $q$ bits (randomly - but not necessarily uniformly chosen) from a given word, and rejects words with probability proportional to their distance from the code. The parameter $q$ is called the locality of the tester.

A random code has, with high probability, constant rate and distance, but locality that is proportional to the length. This is true even for random LDPC codes \cite{Ben-SassonHR05}, and a priori the mere existence of codes with constant locality is not obvious. The first LTCs appear implicitly in works on program checking \cite{BLR} and on probabilistically checkable proofs (PCPs) \cite{BFL,LFKN, BFLS,AS,ALMSS}. A formal definition of an LTC appeared simultaneously in several places  \cite{BFLS,RuSu96,FriedlS13, Arora-thesis} (see \cite{Goldreich2010} for a detailed history). 

Spielman, in his PhD thesis \cite{Spielman96}, discusses the possibility of having an error correcting code that is locally testable (he uses the term `checkable code') and explains its potential applicability: {\em ``A checker would be able to read only a constant number of bits of a received signal and then estimate the chance that a decoder will be able to correct the errors, then the checker can instantly request a retransmission of that block, before the decoder has wasted its time trying to decode the message.
Unfortunately all known codes with local-checkers have rate approaching zero."}

Goldreich and Sudan \cite{GolSud06} initiated a systematic study of LTCs as objects of interest in their own right. Over the years better and better LTCs were constructed \cite{PoliSpiel94, GolSud06, BenSassonSuVaWi03, BGHSV, BenSasson-Sudan05, Din07,  KMRS17, GKORS18}, but, nevertheless, experts went back and forth on whether ``$c^3$-LTCs''  (namely, LTCs with {\bf c}onstant rate, {\bf c}onstant distance, and {\bf c}onstant locality) are likely to exist, compare \cite[Conjecture 3.4]{Goldreich-LTCsurveyOG} with \cite[Section 3.3.2]{Goldreich2010}.

We construct the first such family of LTCs,
\begin{theorem}\label{thm:main}
For every $0<r<1$, there exist $\delta,\kappa>0$ and $q\in \mathbb{N}$ and a polynomial-time construction of an infinite family of error correcting codes $\set{C_n}$ with rate $r$ and distance $\delta$, such that for all $n$, $C_n$ is $\kappa$-locally testable with $q$ queries. 

Namely, every code $C_n$ comes with a randomized local tester that reads at most $q$ bits from a given word $w$ and then accepts or rejects, such that
\begin{itemize}
    \item For all $w\in C_n$, $\Pr[\hbox{accept}]=1.$
    \item For all $w\not\in C_n$, $\Pr[\hbox{reject}] \geq \kappa\cdot \dist(w,C_n)$.
\end{itemize}
\end{theorem}
The parameters $\delta,\kappa$ (and $q$) depend at most  polynomially (and inverse polynomially) on $1-r$, see Remark \ref{rem:params} for more details. 

Remarkably, the theorem actually implies that for every $r,\delta$ for which error-correcting codes are known to exist (approaching the Gilbert-Varshamov bound) there are also LTCs. Indeed, \cite{KMRS17,GKORS18} have shown (see \cite[Section 1.2]{GKORS18}) how to take an LTC with rate arbitrarily close to $1$ and with constant distance, and construct a new LTC with rate and distance approaching the Gilbert-Varshamov bound, and only a constant overhead in the locality $q$. So the theorem above holds for all $r,\delta>0$ that satisfy $r+h(\delta)<1$ where $h(\cdot)$ is the binary entropy function.

\subsection*{Expander codes, one dimension up}
The celebrated expander-codes of Sipser and Spielman \cite{SipserSp96} are a family of error correcting codes constructed from a single base code $C_0\subseteq \bits^d$ and a family of $d$-regular expander graphs $G_n=(V_n,E_n)$ such that the code corresponding to $G_n$ consists of functions on $E_n$ such that for every vertex in $V_n$, the local view at the neighboring edges (assuming some arbitrary fixed ordering) is itself in the base code $C_0$,
\[ C = \sett{f:E_n\to\bits}{\forall v\in V_n, f|_{edges(v)}\in C_0}.
\]
Similarly, our codes will also be defined via a fixed base-code and an infinite family of expander graphs. Our graphs will have, in addition to vertices and edges, also two-dimensional faces, called squares, where each square touches four edges and four vertices. 

Our codewords are functions {\em on the squares} such that for every edge, the bits on the neighboring squares form a codeword in the base code. It is natural to view our code as a Tanner code \cite{Tanner81} with bits on the squares and constraints on the edges; whereas the expander-codes have bits on the edges and constraints on the vertices.

Inspecting our code on the set of squares neighboring a fixed vertex, we see an intermediate code, whose constraints come from the edges neighboring that vertex. 

We thus have three codes for the three dimensions of links: the base code $C_1$ at the link of an edge, the intermediate code $C_0$ at the link of a vertex, and the global code $C$ at the link of the empty face which is the set of all squares.

\paragraph{Left-Right Cayley Complex}
%
Let us describe our construction of a graph-with-squares, namely a square complex (for a more formal description see Definition \ref{def:LRC}). Let $G$ be a finite group with two symmetric sets of generators $A,B$. We define the left-right Cayley complex
$X = Cay^2(A, G,B)$ as follows
\begin{itemize}
    \item The vertices are $X(0)=G$.
    \item The edges are $X(1) = X^A(1)\sqcup X^B(1)$ where
\[X^A(1) = \sett{\set{g,ag}}{g\in G, a\in A}, \qquad X^B(1) = \sett{\set{g,gb}}{g\in
G, b\in B}. \]
\end{itemize}
The fact that with $A$ we multiply on the left, and with $B$ we multiply on the right,  gives a local commutativity which generates many four-cycles, namely, squares.
Indeed for every $a,g,b$ the graph has a cycle of length $4$ with alternating $A$ and $B$ edges, given by the walk $g,gb,agb,ag,g$. We place a square for each of these four-cycles.
\begin{itemize}
     \item The squares are a set of the following four-cycles in the graph,
    \[ X(2) = \sett{(g,gb,agb,ag,g)}{g\in G,a\in A,b\in B}.\] 
    We denote by $[a,g,b]$ the square containing the edges $\set{g,ag}$ and $\set{g,gb}$. By changing the `root' of the square we get $[a,g,b]=[a^{-1},ag,b] = [a^{-1},agb,b^{-1}]= [a,gb,b^{-1}]$.
\end{itemize}

\paragraph{The Code}
Fix a left-right Cayley complex $X=Cay^2(A,G,B)$, and fix a pair of base codes $C_A\subseteq \bits^A$ and  $C_B\subseteq \bits^B$ (assuming $|A|=|B|=d$ we can take both to be isomorphic to some $C_1\subseteq\bits^d$).  Our code is defined to be 
\[ C=
C[A,G,B,C_A,C_B] = \sett{f:X(2)\to \bits}{ \forall a,g,b,\; f([\cdot,g,b])\in C_A,\hbox{ and } f([a,g,\cdot]) \in C_B}.
\]

Observe that for a codeword $f\in C$ and a fixed vertex $g\in G$, the restriction of $f$ to the squares touching $g$ is $f([\cdot,g,\cdot])$. It is not difficult to check that this word necessarily belongs to the tensor code $C_A\otimes C_B$, see Lemma \ref{lem:tensor}. Thus, by putting the constraints around each edge, we get an intermediate code on the squares touching a vertex, which turns out to be a tensor code! Tensor codes have non-trivial dependencies among the constraints defining them. This often implies local testability of tensor codes \cite{BenSasson-Sudan-tensors, DSW06, BVidweakly}, and turns out important for showing that our code $C$ can be locally tested by the following simple test:\\

\textbf{Local test:} Choose a random vertex $g$, and accept iff  $f([\cdot,g,\cdot])\in C_A\otimes C_B$.\\

%

The construction of locally testable codes is completed by describing, in Section \ref{sec:LRCC}, an explicit family of groups and expanding generating sets which give expanding left-right Cayley complexes; and, in Section \ref{sec:ins}, a suitable choice of base codes $C_A,C_B$. \\

Let us now describe how the expansion of the complex facilitates a propagation argument for proving local testability.
\paragraph{Propagation from local to global}
Sipser and Spielman proved distance of their expander codes \cite{SipserSp96} through propagation: expansion of the underlying graph is used to ``lift'' the distance of the base code to the distance of the global code. 
In our codes distance is shown similarly. 

More interestingly, a similar type of argument, but more involved, serves for proving local testability as well. The local testability of the intermediate code $C_A\otimes C_B$ is lifted, via expansion, to imply local testability of the entire code, see Section \ref{sec:alg}. 
\remove{
Indeed,
suppose we are given a word that violates a small amount of constraints. We prove that a simple local correction algorithm (see Algorithm \ref{alg}) converges to a nearby codeword. Why doesn't the local correction algorithm get stuck? There are two components at play. 
\begin{itemize}
    \item The first component is the expansion of the complex, which implies that there will always be vertices that see only a small number of violations in their immediate neighborhood. Here we rely on both the expansion of the underlying graph as well as on the expansion of a certain two dimensional edge-to-square-to-edge walk (called the parallel walk, see Definition \ref{def:par}). 
    \item The second component is the {\em robust testability} of $C_A\otimes C_B$ (this is the so-called intermediate code describing the local view at a vertex). We show that if a vertex sees only a small amount of violated constraints near it, then it can fix the values on squares touching it so that the global amount of violations strictly decreases. 
\end{itemize}
Combining these two components local testability follows almost immediately.
}

We remark that the our code has many dependencies among the constraints defining it. This is to be expected by \cite{Ben-SassonGKSV10}. It is the point where it most clearly differs from expander codes: in expander codes one can have a single violated constraint, which, because it is independent of all other constraints, has no further propagating effect. This leads to a word that is far from the code but cannot be detected by any tester, as proven in  \cite{Ben-SassonHR05}. 

\subsection*{Locally Testable Codes: historical background and techniques}
Let us describe in some more detail the historical background pertaining to locally testable codes, including some works that were already mentioned earlier. 

The study of LTCs arose naturally in works on program checking and PCPs. The Hadamard code was the first code proven to be locally testable in the  work of Blum, Luby, and Rubinfeld on linearity testing \cite{BLR}. The low (logarithmic) rate of this code was quickly improved to polynomial rate by moving from linear functions (codewords of the Hadamard code) to low degree polynomial functions (codewords of the Reed-Muller code). Subsequent works studied ``low degree tests'' which are in fact proofs that  the Reed-Muller code is locally testable. These works were crucial for progress leading up to the proof of the PCP theorem. More on the relation between PCPs and LTCs, as well as the historical development, can be found in Goldreich's survey \cite{Goldreich2010}.



A systematic study of LTCs was initiated by Goldreich and Sudan in \cite{GolSud06}, and a sequence of works constructed both LTCs and PCPs with improved parameters \cite{GolSud06,BenSassonSuVaWi03, BGHSV,Ben-SassonS08,Din07}, achieving constant locality and distance, but rate $1/\poly\log n$. Some experts believed that low rate is inherently needed and some attempts to prove upper bounds on the rate have been made \cite{Ben-SassonGKSV10, DinurK11,ben2012towards, BabaiSS05}, although these lower bounds are in rather restrictive models. 

This, perhaps, has triggered works from the other end of the spectrum \cite{KMRS17,GKORS18} which focused on constructing error correcting codes with constant rate and distance, that are locally testable with smallest possible locality. These works achieve constant rate and quasi-poly-logarithmic distance and locality. 

In terms of techniques, many of the earlier  constructions of LTCs have two notable features. Firstly, they are based on the properties of low degree polynomials, and secondly, they come hand in hand with PCP constructions, so that both share the same composition-recursion structure. 

The gap amplification technique \cite{Din07} of the first author is a construction of both a PCP and an LTC that relies on expander graphs and concatenation and  departs from the domain of low degree polynomials. Meir  \cite{Meir08} gave a tensor-code-based construction of LTCs that is neither related to low degree functions nor to PCPs altogether. Further works  \cite{KMRS17,GKORS18} also construct LTCs without any PCP counterpart. 

A feature shared by all previous constructions of LTCs with mildly high rate is their recursive nature. One first constructs codes with weaker properties and then enhances them by concatenation, possibly with  different iterations. The overall composed structure of the code is somewhat complicated and begs for a more direct ``one-shot'' construction. 

A path leading towards a one-shot construction seemed to open up with the connection to high dimensional expanders.


%


\subsection*{High Dimensional Expansion}
The current paper is mainly elementary and almost self-contained (with the exception of Section \ref{sec:LRCC} which uses the existence of some Ramanujan Cayley graphs with specific properties and can be taken as a black box). But it came up as a result of a much longer and intensive journey.  Some interesting open problems were left aside along the way.  It is, therefore, worthwhile to give the story here.

The journey started by the first and fourth authors during a year-long program at the IIAS (Israeli Institute
of Advanced Studies) on high dimensional expanders in 2017: the hope was to use the Ramanujan complexes (\`a la \cite{LSV1,LSV2}) to construct LTCs as high-dimensional versions of expander-codes over Ramanujan graphs as explained above. 
Although expander codes are typically not locally testable \cite{Ben-SassonHR05} the hope was that higher dimensional versions would be.
    
This optimistic belief was inspired by local to global behavior of certain high dimensional complexes that was uncovered already by Garland in his seminal work \cite{Garland}.
 
In that paper, Garland proved a conjecture of Serre, that the cohomology of co-compact lattices in high-rank simple $p$-adic groups vanishes.  Equivalently,  if $X$ is a finite simplicial quotient of a Bruhat-Tits building of dimension at least two, its cohomology vanishes. The proof of Garland is ``local-to-global'': he showed that if the links of relevant cells have a spectral gap, then so does the global Laplacian of $X$. Namely, if $X$ is locally an expander, then it is also globally so. (For a purely combinatorial treatment and generalizations - see  \cite{Oppenheim18}). The global spectral gap implies the vanishing of the cohomology. 

This ``local to global'' approach is a high-dimensional phenomenon that does not hold for graphs! In graphs, the local structure does not reveal any information about the global expansion. To illustrate this, the reader may recall the LPS-Ramanujan graphs \cite{LPS} which are (p+1)-regular expander graphs with large girth. One can easily get (p+1)-regular graphs with large girth (and hence locally isomorphic to the LPS ones) which are far from being expanders. In contrast, the Garland method shows that local expansion implies global expansion in the high dimensional case.

The local to global approach was also the key ingredient, in \cite{KKL14,EvraK16} where Gromov's overlapping problem was solved using the Ramanjaun complexes.

At this point there was already some interest from the theoretical computer science community. The fact that high dimensional expansion is related to property testing in computer science was observed for the first time by Kaufman and the fourth author \cite{KaufmanL14}. The first author and Kaufman proved that high dimensional expansion implies an efficient agreement-test \cite{DK17}, which is related to both PCPs and LTCs. Anari et al \cite{anari2019log} resolved a conjecture regarding convergence of certain Markov chains by analyzing the global random walk through local analysis at the links.

Inspired by all this, the idea was to construct LTC codes by using the local-to-global behavior of the Ramanujan complexes in an analog to the way \cite{SipserSp96} used Ramanujan graphs for LDPC codes. For simplicity, we will describe it from now on only in dimension $2$, but one can do the same in higher dimensions. 

\remove{As described above, we choose a base code for the highest links, say repeats the above:  the SS code can be thought of as being the subspace of the functions on the edges of a Ramanujan graph whose "local view" at the link of every vertex is in a "small code". If the small code has a good rate and distance the Ramanujan property of the graph enables the propagation of these properties for the "big code". In fact, propagating the rate does not even need the spectral gap.}

The original idea was as follows: fix a large prime $p$ and take an infinite family of Ramanujan complexes $X$, quotients of the Bruhat-Tits building of $G=SL(3,\mathbb{Q}_p)$.  The complex $X$ is a $2$-dimensional complex, the link of every edge of it is in one-to-one correspondence with the projective line $\mathbb{P}^1$ over $\F_p$ and the link of every vertex is the graph of lines versus points of the projective plane over $\F_p$.  One can define a base code (``the small code'') $C_1$ on $\mathbb{P}^1$ to be a  "projective" variant of the Reed-Solomon code. This code induces a "big code" $C$ as a subspace of the $\F_p$ functions on $X(2)$- the $2$-dimensional cells of $X$- whose local views at every edge are in the base code of the edge. The goal was then to propagate the rate, distance, and local testability of Reed-Solomon codes from the small code $C_1$ to the big code $C$. 

This turned out to be easier to say than to do.
At some point, we were hoping to use $p$-adic uniformization.  
Recall the work of Mumford \cite{Mumford} who used the combinatorial structure of one such Ramanujan complex to prove a result on algebraic surfaces appearing as locally symmetric quotients of $SU(2,1)$. We were hoping to go in the opposite direction and to use the theory of algebraic surfaces to study our combinatorial objects. The theory of $p$-adic uniformization was developed in depth by Varshavsky in his thesis \cite{varshavsky1998p} (written under the supervision of the 3rd author of the current paper). This is an opportunity to thank Yakov Varshavsky who gave upon our request a semester-long course describing this work. While we eventually are not using this, we were fortunate to be exposed to an amazing chapter of deep mathematics. 

Propagating local testability from the small code to the big code when these are defined over a high dimensional expander is possible. This was proved in \cite{DiksteinDFH18} with the hope that it would serve our original plan. For our codes to fit, the intermediate code, $C_0$ - the one that is defined on the link of a vertex through the small Reed-Solomon codes $C_1$ on the edges - needed to be itself locally testable. Unfortunately we failed to prove that $C_0$ is locally testable. Here the problem is very concrete: Find $C_1$ inside $\F_p^{\mathbb{P}^1}$ such that the induced intermediate code $C_0$ on the link of a vertex is locally testable. Here, the link of a vertex is nothing but the lines versus points graph of the projective plane. 

One can generalize this challenge to get such a code also on higher dimensional spherical buildings. 
Are such spherical codes locally testable?\\

We, therefore, changed direction and replaced $G=SL(3,\mathbb{Q}_p)$ by a product $G=SL(2,\mathbb{Q}_p)\times SL(2,\mathbb{Q}_q)$.  This time the quotients obtained from congruence lattices in $G$ give rise to square complexes. These complexes were shown long ago to be Ramanujan cubical complexes \cite{JL} and the dynamic of walks along them was studied in \cite{mozes135zero}. 
This time the local intermediate code is a tensor code (since the link of every vertex is the {\em complete} bipartite graph) and there are plenty of tensor codes that are locally testable as mentioned above. A subtle obstacle arose at this point which does not exist in the graph codes of \cite{SipserSp96}: one needs to name the squares in such a way that the function defined on the link of an edge $\set{u,v}$  will be in or out the code independently if we look at it from the vertex $u$ or the vertex $v$.  It might be that this challenge can be overcome, but at that point, we realized that by changing from these square complexes to the left-right Cayley complexes as defined above, this problem is easily fixed. Moreover, it became also easier to argue about the rate- making the whole paper much simpler than we expected! 

  As explained, our long journey left a number of unsettled issues. We believe they are interesting in their
own right (and in all dimensions)  even if not needed anymore for the concrete goal of locally testable codes. 

The left-right Cayley complexes seem to be objects that are worth studying for their own sake.  It is actually somewhat surprising that in spite of over 100 years of studying Cayley graphs, these objects, as far as we know, have never been studied before (but see \cite{BE}, about which we learned only after writing this paper).  An immediate curiosity is whether there are higher-dimensional analogs or whether a group ``has only two sides'' and hence these exist only in dimension $2$. 
Anyway, it seems that this paper solves one problem but opens many others.  \\

After this work was completed and announced, we have learned about related developments which have happened independently in the last few months.

As part of an ongoing effort to build a good LDPC quantum error correcting code (qECC) by homological methods (see \cite{BE,PK1}),  Breuckmann and Eberhardt \cite{BE}, defined a ``balanced product of $G$-graphs''. When specialized to Cayley graphs, this gives the Left/Right Cayley complexes defined in Definition \ref{def:LRC} below.  Using these, some quantum error correcting codes are defined as chain complexes of length $3$. Cutting them to length $2$ gives classical codes. In retrospect, one can see that our codes were hidden there (but without the LTC property).

Even more recently,  Panteleev and Kalachev \cite{PK2} announced solutions for both problems: good quantum LDPC codes  as well as good classical locally testable codes (with rate up to  $1/2$).

\section{Preliminaries}\label{sec:prelim}

\subsection{Expander Graphs}
A $d$-regular graph $G$ is said to be a $\lambda$-one-sided expander if it has eigenvalues $d=\lambda_1 \ge \lambda_2 \ge ... \ge \lambda_n \geq -d$ which satisfy $\lambda_i \leq \lambda\cdot d$ for all $i>1$.

The following is a standard lemma by Alon and Chung,

\begin{lemma}[{\cite{alon1988explicit}}]\label{lemma:AC}
Let $G=(V,E)$ be a $d$-regular $\lambda$-one-sided expander. Let $T\subseteq
V$ be such that the graph induced on $T$, denoted $G(T)$, has average degree at least $\delta d$. Then $|T| \ge (\delta - \lambda )\cdot |V|$, and the number of edges in $G(T)$ is at least $(\delta - \lambda )\delta \cdot|E|$.
\end{lemma} 

This lemma holds in more general situations where instead of a $d$-regular graph we have a weighted Markov operator as long as it has an orthonormal basis of eigenvectors. Let $\D$ be any probability distribution over a finite set $V$, and define an inner product by
\[
\iprod{\cdot,\cdot}_\D:\R^V\times \R^V\to \R, \qquad 
\iprod{f,f'}_\D = \E_{x\sim \D}[f(x)f'(x)].
\]

Let $\ind_T\in \R^V$ be the indicator function of a set $T\subseteq V$. We have $\iprod{\ind_T,\ind_T}_\D=\Pr_\D[T]$, and moreover the probability, with respect to $\D$, that a random walk described by $M$ starts at $T$ and after one step still stays in $T$ is $\iprod{\ind_T,M\ind_T}_\D$. 

Denote by $\one\in\R^V$ the constant $1$ function.
\begin{lemma}\label{lemma:AC-general}
Let $M:\R^V\to \R^V$ be a symmetric Markov operator such that $M \one
= \one$, and such that for all $h$ with $\iprod{h, \one}_\D=0$, $\iprod{h,Mh}_\D\leq
\lambda\iprod{h,h}_\D$. 

Let $f=\ind_T$ be the indicator of a set $T\subseteq
V$. If $\iprod{f,Mf}_\D\geq \delta \cdot \iprod{f,f}_\D$ then $\Pr_\D[T] \ge \delta - \lambda $, and $\iprod{f,Mf}_\D \geq \delta(\delta-\lambda)$.
\end{lemma}

\begin{proof}
Denote $p = \Pr_\D[T]$.
We can write $f = p
\one + h$ with $\iprod{h,\one}_\D=0$. We get
\[\delta \cdot p \leq \iprod{f,Mf}_\D = \iprod{p
\one + h,M(p
\one + h)}_\D \leq
p^2 + \lambda\iprod{h,h}_\D \leq
p^2 + \lambda p  .
\] where the last inequality is because  $\iprod{h,h}_\D\leq\iprod{f,f}_\D = p$. When rearranging, this gives the lemma.
\end{proof}

\subsection{Error Correcting Codes}
A linear code $C\subset \bits^n$ is an $\bits$-linear subspace of $\bits^n$. 
The block-length of the code is $n$. The rate and distance of the code are the relative dimension of the code and relative Hamming weight of the smallest weight non-zero codeword, respectively, namely,
\[
\Rate(C) = \frac 1 n \dim(C) \qquad \mbox{ and } \qquad 
\dist(C) = \frac 1 n \min_{w\in C-\set{0}} |\sett{i\in [n]}{w_i \neq 0}|.
\]

We recall the definition of locally testable codes from \cite{GolSud06}. The definition given here is that of a ``strong'' LTC, and implies all other definitions of locally testable codes. See  \cite[Chapter 13]{goldreich2017introduction}. 

\begin{definition}[Locally Testable Code (LTC)]\label{def:LTC}
For $\kappa>0$ and $q\in \mathbb{N}$ we say that an error correcting code $C\subseteq \bits^n$ is {\em $\kappa$-locally testable with $q$ queries} if there is a distribution over a collection of $q$-element subsets $S\subset [n]$ such that each subset $S$ is associated with a set $V_S\subset \bits^S$ of allowed local views, and such that,  denoting by $f|_S$ the restriction of $f$ to the set $S$, the following hold. 
\begin{itemize}
    \item If $f\in C$ then for every $S$, $f|_S\in V_S$.
    \item For every $f\in\bits^n$, 
\[ \Pr_S [ f|_S\not\in V_S] \geq \kappa\cdot \dist(f,C).
\]
The parameter $\kappa$ is called the {\em detection probability}.
\end{itemize}
\end{definition}

\begin{definition}[Tensor Code]
Let $n_1,n_2\in\N$ and let $C_i \subset \set{f:[n_i]\to\bits}$ for $i=1,2$
be two linear codes. Define the tensor code $C=C_1\otimes C_2$ by 
\[ C = \sett{M:[n_1]\times[n_2]\to \bits}{\forall i\in [n_1],j\in [n_2],
M(i,\cdot)\in C_2, M(\cdot,j)\in C_1}.
\]
\end{definition}

It is easy to check that $\dim(C_1 \otimes C_2)= \dim(C_1) \cdot \dim(C_2)$, and that $\dist(C_1 \otimes C_2) = \dist(C_1)\dist(C_2)$.
We view the elements of $C$ as $n_1$-by-$n_2$ matrices $w$ and write $w(i,\cdot)\in \bits^{n_2}$ for the $i$-th row of $w$, and similarly $w(\cdot,j)\in\bits^{n_1}$ is the $j$-th column of $w$.

A natural {\em test} for whether a given matrix $f\in \bits^{n_1\times n_2}$ is in $C_1\otimes C_2$ is as follows: 

Randomly choose a row or a column, and check whether the restriction of $f$ to that column (or row) is in $C_1$ (or $C_2$). 

The quality of the test is measured by the relation between the rejection probability and the distance of $f$ from the tensor code. Formally, this is captured by the notion of robust testability.

\begin{definition}[Robust testability of tensor codes]
Fix $C_i \subseteq \bits^{n_i}$ linear error correcting codes, for $i=1,2$. For $f:[n_1]\times[n_2]\to\bits$, let 
\[
\dcol(f) = \dist(f, C_1\otimes \bits^{n_2}), \qquad 
\drow(f) = \dist(f,  \bits^{n_1}\otimes C_2).
\] 
and
\[ \delta(f) = (\dcol(f) + \drow(f))/2.\]
The robust testability of $C_1\otimes C_2$ is defined to be
\[ \tau = \min_{f\not\in C_1\otimes C_2} \frac{\delta(f)}{\dist(f,C_1\otimes C_2)}  ,
\]
and we say that $C_1\otimes C_2$ is $\tau$-robustly testable.
\end{definition}

The robust testability of tensor codes was first studied in \cite{BenSasson-Sudan-tensors}, where it was shown that for any code $C$ with sufficiently high distance, the $d$-dimensional tensor code $C^{\otimes d}$ is robustly testable for all $d\geq 3$. The requirement $d\geq 3$ was puzzling because the tensor of Reed-Solomon codes is known  \cite{PoliSpiel94} to be robustly testable even for $d=2$ and this was considered the prototype for locally testable codes. Surprisingly, Paul Valiant discovered \cite{PValiant05} that there are codes $C$ for which $C\otimes C$ is {\em not} robustly testable, see also \cite{GoldreichM07}. Quickly after that \cite{DSW06} formulated a notion of smooth codes, broadened later to `weakly smooth' in \cite{BVidweakly}, and showed that the tensor product of a smooth code and any other code is in fact robustly testable. To define smooth codes recall the definition of LDPC codes,
\begin{definition}[LDPC code]\label{def:ldpc}
Let $c,d,n\in \mathbb{N}$. A $(c,d,n)$-LDPC code is given by a $(c,d)$-regular bipartite graph $([n],[m],E)$ (called a factor graph) with $n$ left vertices and $m=nc/d$ right vertices, called parity checks, such
that all right vertices have degree $d$ and all left vertices have degree
$c$. The code is defined to be 
\[ C = \sett{w:[n]\to\bits}{\forall j\in [m], \sum_{i:ij\in E}w(i) =0 \mod 2}.
\]
\end{definition}

\begin{definition}[Smooth code]\label{def:smooth}
Let $c,d,n\in \mathbb{N},\alpha,\beta,\delta>0$. A $(c,d,n)$-LDPC code $C\subset
\bits^n$ is $(\alpha,\beta,\delta)$-smooth if for every $Y\subseteq [m]$
with $|Y|\leq \alpha \cdot m$ there is some $X\subseteq [n]$ with $|X|\leq
\beta\cdot n$ such that the code $C(\bar Y)|_{\bar X}$ has distance at least
$\delta$, where $\bar Y = [m]\setminus Y$ and $\bar X = [n]\setminus X$. 
Here the code $C(\bar{Y})|_{\bar X}$ is the code obtained by removing
the constraints in $Y$ and then removing the coordinates in $X$.
\end{definition}

Random low density parity check codes (LDPC) are smooth, see Section \ref{sec:ins}. 


\paragraph{Agreement Testability}
A related testing notion focuses on the agreement between pairs of overlapping local views. We think of the following situation,
\begin{itemize}
    \item For each column we are given a codeword of $C_1$, and these are aggregated into $w_1 \in C_1\otimes \bits^{n_2}$.
    \item For each row we are given a codeword of $C_2$, and these are aggregated into $w_2 \in \bits^{n_1}\otimes C_2$.
    \item We check ``agreement'', namely, pick a random pair of row $i$ and column $j$, and check whether they agree on their intersection, i.e. whether 
    \[ w_1(i,j) \stackrel ? = w_2(i,j).\]
\end{itemize}
Agreement testability is defined to be the ratio between the amount of pairwise disagreement to the distance from the code $C_1\otimes C_2$. Formally,
\begin{definition}[agreement testability]\label{def:atest} 
Let $\kappa>0$. Let $C_i \subset \set{f:[n_i]\to\bits}$ for $i=1,2$. We say that $C_1\otimes C_2$ is $\kappa$-agreement testable if for every $w_1 \in C_1\otimes \bits^{n_2}$, $w_2\in \bits^{n_1}\otimes C_2$, there exists $w\in C_1\otimes C_2$ such that
\begin{equation*}
 \kappa \cdot (\Pr_i[w_1(i,\cdot)\neq w(i,\cdot)] + \Pr_j[w_2(\cdot,j)\neq w(\cdot,j)])
\leq \Pr_{i\in [n_1],j\in[n_2]}[w_1(i,j)\neq w_2(i,j)] .
\end{equation*} 
\end{definition}
In words, given a word $w_1$ whose rows are in $C_1$, and given a word $w_2$ whose columns are in $C_2$, we say that $C_1\otimes C_2$ is $\kappa$-agreement testable if the amount of disagreement between $w_1$ and $w_2$ is an upper-bound for the fraction of rows or columns one needs to change in order to get to the closest word $w\in C_1\otimes C_2$, times $\kappa$.

It is well known (see for example \cite{DH09}) that agreement testability is equivalent to robust testability:
\begin{lemma}\label{lem:robagr}
Let $C_i \subseteq\bits^{n_i}$, and assume $\delta_i=\dist(C_i)$ for $i=1,2$. \begin{itemize}
    \item If $C_1\otimes C_2$ is $\tau$-robustly testable then it is  $\kappa$-agreement testable, for $\kappa^{-1} = \frac 1 {2\delta_1\tau} + \frac {1+1/(2\tau)} {\delta_2} $.
    \item If $C_1\otimes C_2$ is $\kappa$-agreement testable, then it is $\tau$-robustly testable for $\tau = \frac\kappa{2(\kappa+1)} $.
\end{itemize}
\end{lemma}
We prove this lemma in Appendix~\ref{app:robagr}.

\section{The Left-Right Cayley Complex}

We describe a new construction of a Cayley graph that in addition to vertices and edges also has two-dimensional faces, called squares. Each square contains four edges that constitute a four-cycle.

\begin{definition}[Left-Right Cayley Complex]\label{def:LRC}
Let $G$ be a group with two symmetric sets of generators $A,B$, namely, each is
closed under taking inverses. We assume that the identity element of $G$ is neither in $A$ nor in $B$. Define the {\em Left-Right Cayley Complex} $X = Cay^2(A, G,B)$ as follows
\begin{itemize}
    \item The vertices are $X(0)=G$.
    \item The edges are $X(1) = X^A(1)\sqcup X^B(1)$ where
\[X^A(1) = \sett{\set{g,ag}}{g\in G, a\in A}, \qquad X^B(1) = \sett{\set{g,gb}}{g\in
G, b\in B}. \]
    \item The squares are $X(2) = A\times G\times B / \sim$ where for every $a\in A,b\in B,g\in G$,
    \[ (a,g,b)\sim (a^{-1},ag,b)\sim (a^{-1},agb,b^{-1})\sim (a,gb,b^{-1}),\]
    and denote the equivalence class of $(a,g,b)$ by $[a,g,b]$, so  \[[a,g,b]=\set{(a,g,b),(a^{-1},ag,b),(a^{-1},agb,b^{-1}), (a,gb,b^{-1})}.\]
\end{itemize}
\end{definition}
The graph $(X(0),X^A(1))$ is none other than the Cayley graph $Cay(G,A)$. Similary $(X(0),X^B(1))$ is the Cayley graph $Cay(G,B)$. The fact that with $A$ we multiply on the left, and with $B$ we multiply on the right,  gives a local commutativity which generates many four-cycles, namely, squares.

\begin{remark}
Given a group $G$ and a set of generators $A$, the Cayley graph $Cay^{left}(G,A)$ with left-multiplication edges is isomorphic to the Cayley graph $Cay^{right}(G,A)$ with right multiplication edges via the map $g\mapsto g^{-1}$.  The left-multiplication edge $\set{g,ag}$ maps to the right multiplication edge $\set{g^{-1}, g^{-1}a^{-1}}$.
This justifies talking about a Cayley graph without specifying left or right multiplication.
\end{remark}
\begin{remark}
The product of two graphs $G_1=(V_1,E_1)$ and $G_2=(V_2,E_2)$ is a square complex $X=G_1\times G_2$ defined as follows. 
\begin{itemize}
\item  The vertices are $X(0)= V_1\times V_2$. \item The edges are $X(1)=E_1\times V_2\;\sqcup \;V_1\times E_2$, where an edge $(\set{u,u'},v)\in E_1\times V_2$ connects $(u,v)$ with $(u',v)$, and similarly an edge  $(u,\set{v,v'})\in V_1\times E_2$ connects $(u,v)$ with $(u,v')$. 
\item The squares $X(2)$ are identified with $E_1\times E_2$, so that the square corresponding to the pair of edges  $e_1=\set{u,u'}\in E_1$ and $e_2=\set{v,v'}\in E_2$ is the four-cycle 
$(u,v)\to (u,v')\to (u',v')\to (u',v)\to (u,v)$. 
\end{itemize}
The left-right Cayley complex is the quotient of the Cartesian product of $G_A = (G,X^A(1))$ and $G_B=(G,X^B(1))$ obtained by identifying the vertex $(g,g')$ with $(gh^{-1},hg')$ for all $h\in G$. One can check that the map $(g,g')\mapsto gg'$ gives a homomorphism of graphs from $G_A\times G_B$ to $Cay^2(A,G,B)$.
\end{remark}
\begin{remark}
Left-right Cayley complexes are examples of two-dimensional cubical complexes. Cubical complexes are well-studied, and in particular there are constructions of Ramanujan cubical complexes \cite{JL} with bounded degree and any dimension, whose walk dynamics was studied in \cite{mozes135zero}. The left-right Cayley complexes have an additional matching labels feature that other complexes are not known to have. 
\end{remark}

\begin{definition}[Links]\label{def:bij}
For each $g\in G$, the link of $g$ is $X_g = \sett{[a,g,b]}{a\in A,b\in B}$. There is a natural  map $(a,b)\mapsto [a,g,b]$. 

For every edge $e=\set{g,ag}$, the link of $e$ is denoted $X_e = \sett{[a,g,b]}{b\in B}$. Similarly, if $e=\set{g,gb}$ we let  $X_e = \sett{[a,g,b]}{a\in A}$.
\end{definition}
\begin{definition}\label{def:TNC}
A left-right Cayley complex satisfies the total no-conjugacy condition if
\[\forall a\in A,b\in B,g\in G,\qquad g^{-1}ag\neq b.\label{eq:nc} \tag{TNC}\]
\end{definition}
Here are a few easy properties of left-right Cayley complexes. %
\begin{claim}\label{claim:NC}
Assuming \eqref{eq:nc}, each vertex has exactly $|A|+|B|$ distinct neighbors; and each square contains exactly four distinct vertices; and the map $(a,b)\mapsto [a,g,b]$ is a bijection from $A\times B$ to $X_g$ for each $g\in G$. 
\end{claim}
\begin{proof}
Let $a\neq a'\in A$ and $b\neq b'\in B$. Clearly $ag\neq a'g$ and $gb\neq gb'$. If $ag = gb$ then $g^{-1}ag = b$ which contradicts \eqref{eq:nc}. So $g$ has $|A|+|B|$ distinct neighbors. 
Next we show that each square $[a,g,b]\in X(2)$ must have four distinct vertices. As $1\not\in A\cup B$, $g\neq ag$ and $g\neq gb$, and we already saw that $ag\neq gb$. Now, if $g = agb$ we would contradict \eqref{eq:nc} because it implies $g^{-1}a^{-1}g=b$ making $a^{-1}\in A$ and $b\in B$ conjugates. 

Finally, let us see that the map $(a,b)\mapsto [a,g,b]$ is a bijection between $A\times B$ and $X_g$ for all $g$. Assume that $[a,g,b]=[a',g,b']$ for some $(a,b), (a',b')\in A\times B$. This implies that $(a',g,b')\in \set{(a,g,b),(a^{-1},ag,b),(a^{-1},agb,b^{-1}), (a,gb,b^{-1})}$. We have seen that $g\neq ag,gb,agb$ hence $(a,g,b)=(a',g,b')$ which means that $(a,b)=(a',b')$. 
\end{proof}
\begin{remark}\label{rem:sizes}
It follows that assuming \eqref{eq:nc} \[|X(1)| = \frac{\card A + \card B}2\cdot |G|\quad \ve\quad |X(2)| = \frac{|A||B|}4 \cdot |G|.\]
\end{remark}

\remove{\begin{claim}\label{claim:expansion}
Assume $Cay(G,A),Cay(G,B)$ are $\lambda$-expanders. Then the $1$-skeleton of $X$, namely the graph $(X(0),X(1))$, is also a $\lambda$-expander.
\end{claim}
\begin{proof}
The graph $(X(0),X(1))$ is a $(|A|+|B|)$-regular graph whose edges are the union of the edges two $\lambda$-expanders on the same vertex set. Let $M_0$ be the Markov operator describing the random walk on $(X(0),X(1))$, and let $M_A,M_B$ be the operators of the random walks on $Cay(G,A)$ and $Cay(G,B)$ respectively. Denote $\alpha = \frac{|A|}{|A|+|B|}$ and $\beta = \frac{|B|}{|A|+|B|}=1-\alpha$, then $M_0 = \alpha M_A + \beta M_B$. Clearly $M_A\one = M_B\one = M_0\one = \one$.
Fix $f:X(0)\to\R$ with $\iprod{f,\one}=0$. We have \[ \iprod{f,M_0 f} = \alpha\iprod{f,M_A f} + \beta\iprod{f,M_B f} \leq \alpha\cdot \lambda \cdot \iprod{f,f}+\beta\cdot \lambda \cdot \iprod{f,f}= \lambda \cdot \iprod{f,f}.\]
\end{proof}
}
It will be useful to consider a weighted version of the $1$-skeleton of $X$, where the weight is distributed evenly between the $A$ and the $B$ edges. When $|A|=|B|$ this is the usual unweighted graph.
\begin{definition}\label{def:distr}
Let $\D_1$ be the distribution over $X(1)$ given by selecting with probability half a uniform edge in $X^A(1)$, and with probability half a uniform edge in $X^B(1)$. (In case $A,B$ have equal size $\D_1$ is the uniform distribution over $X(1)$.)

We define an inner product on functions over $X(1)$. Let $f,f':X(1)\to\R$ and define
\begin{equation}\label{def:iprod}
\iprod{f,f'}_{\D_1} = \E_{e\sim \D_1}[f(e)f'(e)] = \frac 1 2 \E_{e\in X^A(1)}[f(e)f'(e)] + \frac 1 2 \E_{e\in X^B(1)}[f(e)f'(e)].
\end{equation}
This will be the only inner product we consider for functions over $X(1)$ so we sometimes omit the subscript and simply write $\iprod{f,f'}=\iprod{f,f'}_{\D_1}$. As usual we let $\norm f = \iprod{f,f}$.
\end{definition}
\paragraph{Parallel Random Walk}
In addition to the standard random walk on the $1$-skeleton of $X$, we will be interested in a random walk on the edges called  
the parallel walk, which takes an edge $e$ to a random edge $e'$ that is ``parallel'' to it.

\begin{definition}[Labels]
For each $s \in A\cup B$ let $[s] = \set{s,s^{-1}}$. Let $\tilde A = \sett{[a]}{a\in A}$ and let $\tilde B = \sett{[b]}{b\in B}$. The label of an edge $\set{g,ag}$ is defined to be $[a]=\set{a,a^{-1}}$, and this is independent of the presentation of the edge as $\set{g,ag}$ or $\set{(ag),a^{-1}(ag)}$. Similarly, the label of an edge $\set{g,gb}$ is defined to be $[b]=\set{b,b^{-1}}$. 

Let $\tilde A \cup \tilde B$ denote the set of labels of the edges in the complex.
For any $\sigma\in \tilde A \cup \tilde B$, denote by $X^\sigma(1)$ the set of edges labelled $\sigma$.
\end{definition}

\begin{claim}\label{claim:par}
If $\sigma=\set{c,c^{-1}} \in \tilde A \cup \tilde B$ and $c\neq c^{-1}$, then $X^\sigma(1)$ has size $|G|$, otherwise it has size $|G|/2$.
\end{claim}

\begin{proof}
We shall prove the claim for $\sigma=\set{a,a^{-1}} \in \tilde A$, the claim for $\sigma \in \tilde B$ is proven analogously.
Observe that every vertex $g$ participates in two edges labelled $\sigma=\set{a,a^{-1}}$, namely $\set{g,ag}$ and $\set{g,a^{-1}g}$. Since every edge is counted twice, from each of its two endpoints, we get that $|G| = |X^\sigma(1)|$. 

In case $a=a^{-1}$each vertex participates in only a single edge labelled $[a]$, but still every edge has two endpoints so after accounting for the double counting we get $|X^\sigma(1)|=|G|/2$.
\end{proof}

Let us define a Markov operator $\mpar_\sigma:\R^{X^\sigma(1)}\to \R^{X^\sigma(1)}$ on the space of functions on $X^\sigma(1)$. Fix $f:X^\sigma(1)\to\R$. If $\sigma = [a]\in \tilde A$ we set  
\[\mpar_{[a]}f(\set{g,ag}) = \E_b f(\set{gb,agb}),\]
and if $\sigma=[b]\in \tilde B$ we set  
\[\mpar_{[b]}f(\set{g,gb}) = \E_a f(\set{ag,agb}).\]
Note that when $a\neq a^{-1}$, the operator $\mpar_{[a]}$ on $X^{[a]}(1)$ is isomorphic to the standard random walk on $Cay(G,B)$ by sending the edge $\set{g,ag}$ to the vertex $g$. Similarly if $b\neq b^{-1}$ then $\mpar_{[b]}$ is isomorphic to the random walk on $Cay(G,A)$.

We define a Markov operator $\mpar:\R^{X(1)}\to \R^{X(1)}$ on the space of functions on the entire set of edges $X(1)$ by letting, for any $f:X(1)\to\R$, 
\begin{equation}\label{eq:split}
    \mpar f = \sum_\sigma \mpar_\sigma (f|_{X^\sigma(1)}).
\end{equation}
\begin{definition}[Parallel Random Walk]\label{def:par}
We define a random walk on the set of edges $X(1)$ as follows. Starting from an edge $e$, choose uniformly a square containing $e$ and then move to the unique edge $e'\neq e$ on that square with the same label as $e$. (If \eqref{eq:nc} does not hold the square might not contain an edge $e'\neq e$ with the same label, in which case the walk will stay in place).
\end{definition}
The Markov operator corresponding to this walk is exactly  $\mpar$, because starting at an edge $e=\set{g,ag}$, a random square containing $e$ is $[a,g,b]$ for a uniformly chosen $b\in B$, and then the only other $[a]$-labeled edge in this square is the edge $e'=\set{gb,agb}$. 
\begin{lemma}\label{lem:par}
Assume both $Cay(G,A)$ and $Cay(G,B)$ are $\lambda$-expanders. Suppose $R\subseteq X(1)$ and assume $f=\ind_R:X(1)\to\R$ satisfies $\iprod{f,\mpar f} \geq c\cdot \iprod{f,f}$. Then there exists some $\sigma\in \tilde A\cup \tilde B$ such that 
\[ |R\cap X^\sigma(1)| \geq (c-\lambda)|G|/2.
\]
\end{lemma}
\begin{proof}
The weight of $R$ is given by  $\iprod{f,f}$, and the probability to start in $R$ and stay in $R$ after one step of the parallel random walk is $\iprod{f,\mpar f}$.
We expand $\iprod{f,\mpar f}$ according to \eqref{eq:split}, and get 
\[\iprod{f,\mpar f} = \E_\sigma \E_{e\in X^\sigma(1)}[f(e)\mpar_\sigma f(e)],\] 
where the expectation over $\sigma$ is obtained by choosing, with probability half, a random $a\in A$ and outputting $[a]$; and with probability half, a random $b\in B$ and outputting $[b]$. 
Clearly then
\[\iprod{f,f} = \E_{e\sim \D_1}[f(e)^2] = \E_\sigma \E_{e\in X^\sigma(1)}[f(e)^2].\]
Plugging these into the inequality $\iprod{f,\mpar f} - c\cdot \iprod{f,f}\geq 0$ we get 
\[ \E_\sigma \E_{e\in X^\sigma(1)} \left[ {f(e)\mpar_\sigma f(e)} -c\cdot {f(e)^2} \right] \geq 0
\] so there must be at least one $\sigma$ for which 
\begin{equation}\label{eq:par-exp}
    \E_{e\in X^\sigma(1)}[ f(e) \mpar_{\sigma} f(e) ] \geq c \cdot\E_{e\in X^\sigma(1)} 
{[f(e)^2]}. 
\end{equation}
Fix, say, $\sigma = [a]$ and define $h_a:G\to\R$ by $h_a(g) = f(\set{g,ag})$. (The case $\sigma=[b]$ is analogous and omitted). Now,  
\begin{multline}
c \cdot \iprod{h_a,h_a} 
= c\cdot \E_g[h_a(g)^2] 
= c\cdot \E_{g}[f(\set{g,ag})^2] 
= c\cdot \E_{e\in X^{[a]}(1)}[f(e)^2] 
\stackrel{\eqref{eq:par-exp}}\leq \E_{e\in X^{[a]}(1)}[f(e)\mpar_{\sigma} f(e)]
\\ = \E_{g\in G}\left[f(\set{g,ag})\E_{b\in B}[f(\set{gb,agb})] \right]
= \E_{g\in G}\left[h_a(g) \E_{b\in B}h_a(gb)\right] 
= \iprod{h_a,M_B h_a},
\end{multline}
where $M_B$ is the random walk operator on $Cay(G,B)$. We relied here on the fact that choosing a uniform edge in $X^{[a]}(1)$ can be done by choosing a uniform $g\in G$ and looking at $\set{g,ag}$. Observe now that $h_a$ indicates the set $T = \sett{g\in G}{f(\set{g,ag})\neq 0}$, so by Lemma~\ref{lemma:AC-general} applied on the graph $Cay(G,B)$ with the operator $M_B$ we deduce that 
 $|T| \geq (c-\lambda)|G| $. Since every non-zero value for $f$ can cause at most two non-zero values in $h_a$, we get that  $|R\cap X^\sigma(1)|=|f^{-1}(1)\cap X^\sigma(1)|\ge |h_a^{-1}(1)|/2 = |T|/2 
\geq (c-\lambda)\cdot |G|/2$.
\end{proof}%

\section{Error Correcting Code on a Left-Right Cayley Complex}
Let $G, A, B$ and $X=Cay^2(G,A,B)$ as in the previous section. Recall that for any vertex $g \in X(0)$ (resp. any edge $e \in X(1)$) we denote by $X_g \subset X(2)$ (resp. $X_e \subset X(2)$) the set of squares in $X$ containing the vertex $g$ (resp. the edge $e$). Let $C_A\subset \bits^A$ and let $C_B\subset \bits^B$ be two fixed linear error correcting codes with rates $\rho_A = \Rate(C_A),\rho_B = \Rate(C_B)$ and distances $\delta_A = \dist(C_A),\delta_B = \dist(C_B)$, respectively. 

Define the code $C=C[G,A,B,C_A,C_B]$ as follows.
For an edge $e=\set{g,ag}\in X^A(1)$, we define a local code 
\[ C_e = \sett{f:X_e\to\bits}{f([a,g,\cdot]) \in C_B}.
\]
Similarly, for an edge $e=\set{g,gb}\in X^B(1)$, we define a local code 
\[ C_e = \sett{f:X_e\to\bits}{f([\cdot,g,b]) \in C_A}.
\]
Note that this definition appears to depend on the choice of $g\in e$ but it does not. 
Finally, we define a global code 
\[ C = \sett{f:X(2)\to \bits}{\forall e\in X(1), f|_{X_e}\in C_e}.
\]

For each vertex $g \in X(0)$, define the local tensor code around the vertex $g$ to be
\[ C_g = \sett{f:X_g\to \bits}{f([\cdot,g,\cdot]) \in C_A\otimes C_B}.
\] 

\begin{lemma}[$C$ is a lifted tensor-code]\label{lem:tensor}
\[ 
C = \sett{f:X(2)\to \bits}{\forall g\in X(0), f|_{X_g}\in C_g}.
\]
\end{lemma}
\begin{proof} Immediate from the fact that $f([\cdot,g,\cdot]) \in C_A\otimes C_B$ for any $g\in X(0)$ if and only if $f([a,g,\cdot]) \in C_B$ and $f([\cdot,g,b]) \in C_A$ for any $g\in X(0)$, $a\in A$ and $b \in B$.

\end{proof}
Observe that for the local code at each vertex to be a tensor code, we must make sure that around every $A$ edge we have {\em the same} code $C_A$, and similarly for $B$. If we choose different base codes at different edges we might still get a code with rate and distance, but local testability will probably fail, because we lose the local tensor structure. This is in contrast to the case of expander codes where the local base code can be chosen arbitrarily and differently at each vertex.

\subsection{Properties of The Code}

We now look at the rate, distance and local testability of the code $C = C[G,A,B,C_A,C_B]$. 
Recall $\rho_A = \Rate(C_A)$, $\rho_B = \Rate(C_B)$ and $\delta_A = \dist(C_A)$, $\delta_B = \dist(C_B)$

\begin{lemma}[Rate]\label{lem:rate}
The rate of the code $C$ is bounded from below by 
\[
\Rate(C) \geq 2(\rho_A+\rho_B) - 3.
\]
\end{lemma}

\begin{proof}
For each $e\in X^A(1)$, $codim(C_e) = codim(C_B) = |B|\cdot(1-\rho_B)$. Similarly
for each $e\in X^B(1)$, $codim(C_e) = codim(C_A) = |A|\cdot(1-\rho_A)$. The
number of linearly independent constraints on $f\in C$ is at most
\[
|X^A(1)|\cdot |B| (1-\rho_B) + |X^B(1)|\cdot |A| (1-\rho_A)  = |G||A||B|(1 - \frac {\rho_A+\rho_B}2)
\]
On the other hand, the dimension of the ambient space is the number of squares $|X(2)| = |G||A||B|/4$, see Remark \ref{rem:sizes}. 
Subtracting the number of constraints from the number of bits we get a lower bound on the dimension of the code,
\[
dim(C)\geq \frac{1}{4} |G||A||B|(1 - (4 - 2(\rho_A+\rho_B))) = \frac{1}{4} |G||A||B|(2(\rho_A+\rho_B)-3).
\]
\end{proof}

In fact, we can do a little better. Recall that a {\em vertex cover} of a graph is a set of vertices that touch all of the edges. For example, if the graph is bipartite, then it has a vertex cover whose size is half the size of the graph.

\begin{lemma}[Rate - better bound]
Suppose the underlying graph of $X$ has a vertex cover of size $\nu|G|$. Then the rate of the code is at least $4\nu\rho_A\rho_B +1-4\nu$. 
In particular, if the graph is bipartite, then $\nu=\frac 1 2$ and we get that 
\[
\Rate(C) \geq 2 \rho_A \rho_B - 1.
\]
\end{lemma}
It is interesting to mention that in the expander codes of Tanner \cite{Tanner81}, (whose distance and decoding were later analyzed in \cite{SipserSp96}), if the local code $C_0$ has rate $\rho_0$ then the global rate is shown to be at least $2\rho_0-1$. In our code the rate of the local code is $\Rate(C_g) = \Rate(C_A\otimes C_B) = \rho_A\rho_B$, and in case the graph is bipartite, we get the same bound of $2(\rho_A\rho_B)-1$ on the rate of the global code. 

\begin{proof}
Let $V^*\subset G$ be a vertex cover, namely, a set of vertices that touches every edge. Then $f\in C$ if and only if for every $g\in V^*$, $f|_{X_g}\in C_g$. The reason is that every edge $e$ touches some $g\in V^*$ and the constraint $f|_{X_e}\in C_e$ is implied by $f|_{X_g}\in C_g$. 

Since $C_g$ is isomorphic to $C_A\otimes C_B$ it has $|A|\cdot|B|(1 -\rho_A\rho_A)$ linearly independent constraints. The dimension of the code is at least 
\begin{multline} 
\dim(C) \geq  |G||A||B|\frac 1 4 - |V^*|\cdot |A|\cdot|B|(1 -\rho_A\rho_A) \\
\geq \frac 1 4 |G||A||B|\cdot(1 - 4\nu(1-\rho_A\rho_B))
= \frac 1 4 |G||A||B|\cdot (4\nu\rho_A\rho_B +1-4\nu).
\end{multline}
\end{proof}

\begin{lemma}[Distance] \label{lem:distance}
Suppose that both $Cay(G,A)$ and $Cay(G,B)$ are $\lambda$-expanders for $\lambda<1$.
Then the distance of the code $C$ is bounded from below by
\[
\dist(C) \geq \delta_A \delta_B \cdot \left( \max(\delta_A,\delta_B) - \lambda \right).
\]
\end{lemma}

\begin{proof}
Let $0\neq f\in C$.
Let $g_0\in X(0)$ be some vertex such that $w_{g_0} = f|_{X_{g_0}} \neq 0$ (if they are all zero then $f=0$). Observe that since $0\neq w_{g_0}\in C_A\otimes C_B$ then $w_{g_0}$ has at least $\delta_A|A|$ non zero columns 
and at least $\delta_B|B|$ non-zero rows. Let $A_1\subset A$ be the labels of these columns, and fix $a_1\in A_1$. We first show that 
\begin{equation}\label{eq:rows}
\Pr_{g,b}[f([a_1,g,b])\neq 0] \ge \delta_B(\delta_B-\lambda).
\end{equation} 

To prove \eqref{eq:rows} consider the graph $Cay(G,B)$ whose vertices are $X(0)$ and the edges are $X^B(1)$, and define a function $f_{a_1}:X^B(1) \to\bits$ by $f_{a_1}(\set{g,gb}) = f([a_1,g,b])$. Observe that $f_{a_1}$ is well defined because for $g' = gb$, 
\[
f_{a_1}(\set{g,g'}) = f_{a_1}(\set{g,gb})=f([a_1,g,b])=f([a_1,g',b^{-1}])= f_{a_1}(\set{g',g'b^{-1}}) = f_{a_1}(\set{g',g}).
\] 

Since $f_{a_1}\neq 0$, it must have large weight because it belongs to the expander code defined on $Cay(G,B)$ with local code $C_B$. More elaborately, for every vertex $g$ that touches an edge where $ f_{a_1}\neq 0$, there must be at least $\delta_B|B|$ non-zero edges touching $g$. By Lemma \ref{lemma:AC} we get at least $\delta_B(\delta_B-\lambda)|X^B(1)|$ edges on which $f_{a_1}\neq 0$, which proves \eqref{eq:rows}. 

For every $a\in A_1$, the weight of $f_{a}$ is at least $\delta_B(\delta_B-\lambda)$, so if we choose a random $a\in A$ and then a random edge in $X^B(1)$, the probability that $a\in A_1$ is at least $\delta_A$, and conditioned on this, the probability that $f_{a}(e)\neq 0$ is at least $\delta_B(\delta_B-\lambda)$, so altogether
\[ \Pr_{a,g,b}[f([a,g,b])\neq 0] \geq  \Pr_a[a\in A_1]\cdot \Pr_{g,b}[f_{a}(\set{g,gb})\neq 0\;|\; a\in A_1] \geq \delta_A   \delta_B(\delta_B-\lambda).
\]
Symmetrically, the weight of $f$ is also at least $\delta_B\delta_A(\delta_A-\lambda)$, and the lemma follows.
\end{proof}

\begin{theorem}[Local Testability]\label{thm:mainltc}
Suppose $X=Cay^2(A,G,B)$ is a left-right Cayley complex such that both $Cay(G,A)$ and $Cay(G,B)$ are $\lambda$-expanders, and such that \eqref{eq:nc} holds. Assume $C_A\subset \bits^A$ and $C_B\subset \bits^B$ are error correcting codes with relative distances $\delta_A,\delta_B>0$ respectively and such that $C_A\otimes C_B$ is $\kappa_0$-agreement testable.
If 
\begin{equation}\label{eq:c}
    c=\frac{ \kappa_0}{8+\kappa_0} \cdot \min(\delta_A,\delta_B) >\lambda
\end{equation}
 then $C = C[G,A,B,C_A,C_B]$
is $\kappa$-locally testable with $|A|\cdot|B|$ queries, 
where \[\kappa = \min\left(\frac{1}{4(1+|A|+|B|)}, \frac{c-\lambda}{2(|A|+|B|)}\right).\]  
Namely, for every $f:X(2)\to\bits$,
\[  \Pr_{g\in X(0)}[f|_{X_g}\not \in C_g] \geq \kappa \cdot \dist(f,C).
\] 
\end{theorem}

In words, given some potential codeword $f$, each vertex $g$ is associated with a local test that reads $f$ at all of the $|A|\cdot|B|$ squares touching $g$ and checks that these values form a codeword in the base code $C_g \cong C_A\otimes C_B$. The theorem says that the distance of $f$ to the code is upper bounded by a constant multiple of the fraction of violated local tests. 

Given a base code with distance $\delta_0$ and agreement testability $\kappa_0$, the testability of the entire code is about  $\Omega(\delta_0\kappa_0/(|A|+|B|))$.

We prove the theorem in the next section, by describing an iterative correction algorithm that finds a codeword close to $f$ if the probability that the test rejects is not too large.

\subsection{Local Self-Correction Algorithm}\label{sec:alg}

In this section we describe a local self-correction algorithm, see Algorithm \ref{alg}, that starts with a given string $f:X(2)\to\bits$ and either finds a codeword $f_0\in C$ or gives up. We denote \[\vrej(f) = \Pr_{g}(f|_{X_g}\not\in C_{g}),\]
the fraction of rejecting local tests. We will show that if $\vrej(f)\leq \vrej_0$ for some constant $\vrej_0>0$, then the algorithm finds $f_0\in C$ such that $\dist(f_0,f)\leq O(\vrej(f))$.  

For each vertex $g$, let $\L_g\in C_g$ be a closest codeword to $f|_{X_g}$ (breaking ties arbitrarily).
We focus on the collection of local views $\W = \set{\L_g}$ and whether the local views of neighboring vertices agree on the common squares.
\begin{definition}
Given a collection of local views $\W=\sett{\L_g\in C_g}{g\in G}$, we define the disagreement of the collection to be 
\begin{equation}\label{eq:Eps}
\disagr(\W) = \Pr_{e=\set{g,g'}\in X(1)}[\L_g|_{X_{e}} \neq \L_{g'}|_{X_{e}}] \end{equation}
where $e$ is a uniformly random edge in $X(1)$.
\end{definition}
\RestyleAlgo{ruled}
\begin{algorithm}\label{alg}
\begin{enumerate}
    \item (Initialization:) For each vertex $g$, let $\L^0_g\in C_g$ be a closest codeword to $f|_{X_g}$ (breaking ties arbitrarily). \[ \L^0_g = argmin_{w\in C_g} \dist(w,f|_{X_g}).\]
Let $\L_g \leftarrow \L^0_g$ for all $g\in G$, and let $\W=\set{\L_g}$.
    \item (Main loop:) If there is a vertex $g$ and a choice $w\in C_g$ that reduces $\disagr(\W)$ then replace $\L_g$ by $w$ and repeat. 
    \item (End:) If $\disagr(\W)>0$ output ``far''. Otherwise, $\disagr(\W)=0$, define $f_0:X(2)\to\bits$ by choosing for each square $s\in X(2)$ an arbitrary  vertex $g\in s$ and setting $f_0(s) = \L_g(s)$. Output $f_0$.
\end{enumerate}
\caption{Iterative decoding algorithm. (input: $f:X(2)\to\bits$)}
\end{algorithm}

Observe that $\disagr(\W)|X(1)|$ is a non-negative integer, and this value decreases by at least $1$ every step of the algorithm, so the algorithm must halt. 
\begin{proposition}\label{prop:completeness}
 If the algorithm outputs $f_0$ then $f_0\in C$ and 
 \[ \dist(f,C)\leq \dist(f,f_0) \leq 4(1+|A|+|B|)\cdot \vrej(f).
 \]
\end{proposition}
Let $\W^0 = \set{\L^0_g}$ be the collection of local views defined in the initialization step of the algorithm, and let $\W=\set{\L_g}$ be the final collection, at the end of the algorithm. 
\begin{proposition}\label{prop:soundness}
 If the algorithm outputs ``far'' then $\disagr(\W)\geq \eps_0 = \frac{c-\lambda}{|A|+|B|}$, where  $c=\frac{ \kappa_0}{8+\kappa_0} \cdot \min(\delta_A,\delta_B)$ is defined in \eqref{eq:c}.
\end{proposition}
\remove{We will show that this immediately means that $\vrej(f)\geq \disagr(\W)/2\geq \frac{c-\lambda}{2(|A|+|B|)}$ and this in turn means that $\dist(f,C)\leq 1 \leq \frac{2(|A|+|B|)}{c-\lambda}\cdot \vrej(f),$ which will prove the theorem.}
\begin{proof}[Proof of Theorem~\ref{thm:mainltc} assuming Propositions \ref{prop:completeness} and \ref{prop:soundness}]
 Given $f:X(2)\to\bits$, run the algorithm above. The output is either a function $f_0$, which by Proposition~\ref{prop:completeness}, satisfies $\dist(f,C)\leq \dist(f,f_0) \leq 4(1+|A|+|B|)\cdot \vrej(f)$; or the output is ``far'', in which case $\disagr(\W)\geq\eps_0$ by Proposition~\ref{prop:soundness}. We observe that 
\begin{equation}\label{eq:rho}
    \disagr(\W^0) \leq 2\vrej(f) 
\end{equation}
for the followin reason. For each edge $\set{g,g'}$ that contributes to $\disagr(\W^0)$ either $f|_{\T g} \neq \L^0_g$
or $f|_{\T {g'}} \neq \L^0_{g'}$, otherwise \[\L^0_g|_{\T{gg'}} = (f|_{\T {g}})|_{\T{gg'}}
= f|_{\T {gg'}} = (f|_{\T {g'}})|_{\T{gg'}}= \L^0_{g'}|_{\T{gg'}}.\]
Therefore, the process of selecting an edge uniformly and then a random endpoint of it will lead to a rejecting vertex with probability at least $\disagr(\W^0)/2$, proving \eqref{eq:rho}. 

Now $\vrej(f) \geq \disagr(\W^0)/2 \geq \disagr(\W)/2 \geq \eps_0/2 = \frac{c-\lambda}{2(|A|+|B|)}$, so we can write
 \[\dist(f,C)\leq 1\leq \frac{2(|A|+|B|)}{c-\lambda}\cdot\vrej(f).\] 
 All in all we get,
 \[ \dist(f,C) \leq \max(4(1+|A|+|B|),\frac{2(|A|+|B|)}{(c-\lambda)})\cdot \vrej(f) = \kappa \cdot \Pr_{g}(f|_{X_g}\not\in C_{g})
 \] as needed.
\end{proof}
 \begin{remark}
 Algorithm \ref{alg} is clearly also a decoding algorithm in the standard sense: if we know that the given word $f$ is close enough to the code, then the regular structure of the tester (each square affects exactly four vertices) implies that it will be rejected with probability proportional to $\dist(f,C)$. The analysis herein shows that for small enough (constant) distance, the algorithm will then find the nearest codeword.
 \end{remark}
We now turn to prove the two propositions.
\begin{proof}[Proof of Proposition \ref{prop:completeness}]
By assumption, $\disagr(\W)=0$. We first observe that the value of $f_0(s)$ does not depend on the choice of $g\in s$ because $\disagr(\W)=0$ implies that $\L_{g}(s) = \L_{g'}(s)$ for any
$g,g'\in s$. (Suppose $g_1,g_2\in s$ disagree. If they are adjacent this
means that $\L_{g_1}$ disagrees with $\L_{g_2}$ contradicting $\disagr(\W)=0$.
If they are  non-adjacent, they have a common neighbor which cannot agree with both of them). It follows that $f_0\in C$, because for each $g$, $f_0|_{X_g}=\L_g\in C_g$.
To bound $\dist(f,f_0)$, let \[
V_0 = \sett{g\in  X(0)}{f|_{\T {g}}\neq \L^0_g},\qquad 
V_1 = \sett{g\in X(0)}{\L^0_g\neq \L_g}.\] 
So $V_0$ is the set of vertices whose local view does not perfectly satisfy the constraints of the code, and $V_1$ is the set of vertices $g$ for which $\L_g$ at the end of the algorithm differs from its initial value.

Observe that $g\in V_0$ iff $f|_{\T g}\not\in C_g$, so by definition,
\begin{equation}\label{eq:V0}
|V_0| = \vrej(f)\cdot |X(0)|.    
\end{equation}
Any square $s$ that does not touch $V_0\cup V_1$ must have for every $g\in s$ 
\[f_0 (s) = \L_g(s) = \L_g^0(s) = f(s),\]
where the second equality is because $g\not \in V_1$ and the third is because $g\not \in V_0$.
We bound $|V_1|$ by the number of iterations of the algorithm, which is at most $|V_1|\leq \disagr(\W^0) \cdot |X(1)|$. We recall from \eqref{eq:rho} that $
    \disagr(\W^0) \leq 2\vrej(f)$. 
Thus, we have,
\begin{equation}\label{eq:V1}
|V_1|\leq \disagr(\W^0) \cdot |X(1)| \leq 2\vrej(f)\cdot \frac  {|A|+|B|}2 |X(0)|.    
\end{equation}
Altogether, since every vertex touches $|A||B|$ squares, and since $|X(2)| = |A||B||X(0)|/4$, and using \eqref{eq:V0} and \eqref{eq:V1}, we get
\[ \dist(f, f_0) \leq  \frac {|A||B| \cdot |V_0\cup V_1| }{|X(2)|} =
 \frac {4 \cdot |V_0\cup V_1| }{|X(0)|} \leq
 4(1+ {|A|+|B|})\vrej(f) .
\]
\end{proof}

The interesting part of the proof is to show that if $\disagr(\W)>0$ after
the algorithm ends, then $\disagr(\W) > \eps_0 = \frac {c-\lambda}{|A|+|B|} $.
\begin{proof}[Proof of Proposition~\ref{prop:soundness}]
Let
\[R = \sett{e=\set{g,g'}\in X(1)}{\L_g|_{\T{e}}\neq  \L_{g'}|_{\T{e}}}\] be the set of ``dispute'' edges. 
The rest of the proof is aimed towards showing $\disagr(\W)\geq \eps_0$ or equivalently, since $\disagr(\W) = |R|/|X(1)|$, that
\begin{equation}\label{eq:R} |R| \geq \frac{c-\lambda}{|A|+|B|}\cdot|X(1)| =\frac{c-\lambda}{2} \cdot |G|.
\end{equation}

First, some more notations.
For an edge $\set{g,ag}\in X^A(1)$ let 
\[\gpar(\set{g,ag}) = \sett{\set{gb,agb}\in X^A(1)}{b\in B}\] and similarly
for an edge $\set{g,gb}\in X^B(1)$, \[\gpar(\set{g,gb}) = \sett{\set{ag,agb}\in
X^B(1)}{a\in A}.\]
For a vertex $g$, let 
\[ \ga(g) = \sett{\set{g,ag}}{a\in A},\qquad \gb(g) = \sett{\set{g,gb}}{b\in
B}. \]

We now make two claims on the local structure of $R$. The first is due to
the local distance, and the second is due to the local testability of our tensor
code.
\begin{claim}\label{claim:s1}
Suppose $\set{g,ag}\in R$, then 
\[ |R\cap \gb(g)| + |R\cap \gb(ag)| + |R\cap\gpar{\set{g,ag}}|\geq \delta_B|B|.\]
Similarly, suppose $\set{g,gb}\in R$, then 
\[ |R\cap \ga(g)| + |R\cap \ga(gb)| + |R\cap\gpar{\set{g,gb}}|\geq \delta_A|A|.\]
\end{claim}
\begin{proof}Let $e=\set{g,ag}\in R$, so $\L_g|_{X_e}\neq \L_{ag}|_{X_e}$. Since $\L_g|_{X_e},\L_{ag}|_{X_e}\in C_e$, these are two distinct codewords of $C_e$, and must disagree on at least $\delta_B|B|$ squares. Let $[a,g,b]$ be such a square, and look at the three edges of the square that are not $e$: $\set{g,gb},\set{gb,agb}$ and $\set{agb,ag}$. At least one of the three edges must be in $R$, because $\L_g,\L_{gb},\L_{agb},\L_{ag}$  cannot all agree on the value of $[a,g,b]$ without contradicting $\L_g([a,g,b]) \neq \L_{ag}([a,g,b])$. This implies the first part of the claim, and the second part is proven similarly.
\end{proof}
Recall that we assume $C_A\otimes C_B$ is agreement testable, as per Definition \ref{def:atest}. \begin{claim}\label{claim:s2}
Assume $C_A\otimes C_B$ is $\kappa_0$-agreement testable.  For every $g\in G$,
\begin{equation}\label{eq:ltc}\Pr_a[\set{g,ag}\in R]+\Pr_b[\set{g,gb}\in R]\leq \kappa_0^{-1}\cdot \Pr_{a\in A,b\in B}[\set{ag,agb}\in R\hbox{ or }\set{gb,agb}\in R
].
\end{equation}
\end{claim}
\begin{proof}
Define $w_0,w_1,w_2:A\times B\to\bits$ as follows. First, let $w_0(a,b) = \L_g([a,g,b])$.
Next, let
$w_1(a,b) = \L_{ag}([a^{-1},ag,b])$. Similarly let $w_2(a,b)=\L_{gb}([a,gb,b^{-1}])$.
In words, the $a$th row of $w_1$ comes from the ``opinion'' of $\L_{ag}$,
and the $b$th column of $w_2$ comes from the ``opinion'' of $\L_{gb}$. Observe
that $w_0\in C_A\otimes C_B$, $w_1\in \bits^A\otimes
C_B$, and $w_2 \in C_A\otimes \bits^B$.
Now observe that $w_1(a,\cdot)\neq w_0(a,\cdot)$ iff $\set{g,ag}\in R$, and
$w_2(\cdot,b) \neq w_0(\cdot,b)$ iff $\set{g,gb}\in R$. Finally, $w_1(a,b)\neq
w_2(a,b)$ implies that the event on the RHS of \eqref{eq:ltc} holds, namely, $\set{ag,agb}\in
R\hbox{ or } \set{gb,agb}\in R$. 

By the $\kappa_0$-agreement testability of $C_A\otimes C_B$, there is a word $w\in C_A\otimes C_B$ such that
\[ \Pr_a[w(a,\cdot)\neq w_1(a,\cdot)] + \Pr_b[w(\cdot,b)\neq w_2(\cdot,b)] \leq \kappa_0^{-1}\cdot \Pr_{a,b}[w_1(a,b)\neq w_2(a,b)].
\]
Since the iterative algorithm has terminated, we know that
\[
\Pr_a[w_0(a,\cdot)\neq w_1(a,\cdot)] + \Pr_b[w_0(\cdot,b)\neq w_2(\cdot,b)] \leq  \Pr_a[w(a,\cdot)\neq w_1(a,\cdot)] + \Pr_b[w(\cdot,b)\neq w_2(\cdot,b)]
\]
otherwise the algorithm would have flipped from $\L_g=w_0$ to $\L_g=w$. Combining the inequalities the claim follows,
\begin{align*}
\Pr_a[\set{g,ag}\in R]+\Pr_b[\set{g,gb}\in R] &=
\Pr_a[w_0(a,\cdot)\neq w_1(a,\cdot)] + \Pr_b[w_0(\cdot,b)\neq w_2(\cdot,b)] \\
&\leq
\Pr_a[w(a,\cdot)\neq w_1(a,\cdot)] + \Pr_b[w(\cdot,b)\neq w_2(\cdot,b)] \\
&\leq \kappa_0^{-1}\cdot \Pr_{a,b}[w_1(a,b)\neq w_2(a,b)]\\
&\leq \kappa_0^{-1}\cdot \Pr_{a\in A,b\in B}[\set{ag,agb}\in R\hbox{ or }\set{gb,agb}\in R].
\end{align*}
\end{proof}
Let $M_0 = \frac 1 2 M_A+\frac 1 2 M_B$, where $M_A,M_B$ are the operators of the random walks on $Cay(G,A)$ and $Cay(G,B)$ respectively. Clearly for any $f:X(0)\to\R$ such that $E[f]=0$, $\iprod{f,M_0 f} = \frac 1 2 \iprod{f,M_A f}+\frac 1 2 \iprod{f,M_B f} \leq \lambda\iprod{f,f}$. 
Recall the distribution $\D_1$ over $X(1)$ from Definition~\ref{def:distr} and the corresponding inner product $\iprod{\cdot,\cdot}_{\D_1}$.
Define $\Down:\R^{X(1)}\to\R^{X(0)}$, $\Up:\R^{X(0)}\to\R^{X(1)}$ to be the down and up operators, moving us from functions on edges to functions on vertices and vice versa. Namely, 
\[\forall f_1\in \R^{X(1)},\qquad  \Down f_1(g) = \E_{e\sim\D_1|g}[f_1(e)]= \frac 1 2 \E_{a\in A}[f_1(\set{g,ag})]+ \frac 1 2 \E_{b\in B}[f_1(\set{g,gb})]\] and \[\forall f_0\in \R^{X(0)},\qquad \Up f_0(\set{g_1,g_2}) = \E_{g\in \set{g_1,g_2}} [f_0(g)]=\frac 1 2(f_0(g_1)+f_0(g_2)).
\]
Note that these are averaging operators so they never increase norms, e.g. $\norm{\Down f} \leq \norm f$ for all $f$. 
\begin{claim}
Let $M=\Up M_0\Down:\R^{X(1)}\to \R^{X(1)}$. Then $M$ has second largest eigenvalue at most $\lambda$.
\end{claim}
\begin{proof} 
We rely on the fact that $\D_1$ can be described by first choosing a uniform vertex $g$ and then a random edge containing $g$ such that with probability half we choose an $A$ edge and with probability half a $B$ edge.
For any $f_1:X(1)\to\R$ and $f_0:X(0)\to\R$ we have 
\[ \iprod{\Down f_1,f_0} = \E_{g}[\E_{e\sim\D_1|g} [f_1(e)]\cdot f_0(g)] = \E_{e\sim \D_1} [f_1(e)\E_{g\in e} [f_0(g)]] = \iprod{f_1,\Up f_0}_{\D_1}.
\] Now, if $\iprod{f_1,\one}=0$ then $\iprod{\Down f_1,\one}=0$, so
\[ \iprod{f_1,Mf_1}=\iprod{f_1,\Up M_0 \Down f_1} = \iprod{\Down f_1,M_0 \Down f_1} \leq \lambda\iprod{\Down f_1,\Down f_1}\leq \lambda\iprod{f_1,f_1}.
\]
\end{proof}
The following lemma is based on Claims \ref{claim:s1} and \ref{claim:s2}.
\begin{lemma}\label{lemma:stay} Fix $\gamma = \frac{\kappa_0}{8+\kappa_0}$. Let $M = \Up M_0 \Down$ and let $f=\ind_R:X(1)\to\R$ be the indiator function of the edge set $R$. Then 
\[ \iprod{f,(\gamma\mpar +(1-\gamma) M)f}_{\D_1} \geq \gamma \cdot\min(\delta_A, \delta_B) \cdot \iprod{f,f}_{\D_1}.
\] 
\end{lemma} 
\def\EBA{\mathsf{E}_{BA}}
\def\EAB{\mathsf{E}_{AB}}
\def\Et{\mathsf{E}_{3}}
\begin{proof}
We give a combinatorial interpretation to $\gamma\mpar +(1-\gamma) M$ by observing that for a fixed $e\in X(1)$, $(\gamma\mpar +(1-\gamma) M)f(e)$ is the probability that $e'\in R$ in the following random process.
\begin{enumerate}
    \item Start from an edge $e\in X(1)$.
    \item\label{step:par} With probability $\gamma$, output a uniformly random edge $e'\in \gpar(e)$ and halt. With probability $1-\gamma$ continue.
    \item\label{step:down} Choose at random one of the endpoints of the edge, $g_1 \in e$.
    \item With probability $\frac 1 2$ let $g_2 = a_1g_1$ for a random $a_1\in A$, and with probability $\frac 1 2$ let $g_2= g_1b_1$ for a random $b_1\in B$. 
    \item With probability $\frac 1 2$ let $e'=\set{g_2, a_2g_2}$ for a random $a_2\in A$, and with probability $\frac 1 2$ let $g_2= g_2b_2$ for a random $b_2\in B$. Output $e'$.
\end{enumerate}
We will prove the lemma by showing that for every $e\in R$, 
\begin{equation}\label{eq:stay}
    (\gamma\mpar +(1-\gamma) M)f(e)\geq \gamma\cdot\min(\delta_A,\delta_B).
\end{equation}
So fix some $e \in R$, and for convenience assume $e = \set{g,ag}$ for some $g\in G,a\in A$ (if $e=\set{g,gb}$ the argument is symmetric). Let 
\[r_0=|R\cap \gpar(e)|,\quad r_1 =|R\cap \gb(g)|,\quad r_2=|R\cap \gb(ag)|.\]  
By Claim \ref{claim:s1}, $r_0+r_1+r_2  \geq \delta_B|B|$. With probability $\gamma$ step \ref{step:par} outputs a random $e'\in \gpar(e)$, and the probability it is in $R$ is $r_0/|B|$.
\begin{equation}\label{eq:step2}
\Pr[e'\in R] = \gamma \cdot r_0/|B| + (1-\gamma)\cdot \Pr[e'\in R\,|\,\hbox{the process entered step 3}]    
\end{equation}

Assume we entered step \ref{step:down}. Due to Claim \ref{claim:s2}, 
\begin{equation}\label{eq:edge}
    \Pr_{a,b}[\set{ag_1,ag_1b}\in R \hbox{ or }\set{g_1b,ag_1b}\in R] \geq \kappa_0\cdot r_i/|B|
\end{equation} 
where $i\in\set{1,2}$ depending on whether $g_1=g$ or $g_1=ag$ as chosen in step \ref{step:down}. 
What is the probability that $e'$ is one of the edges $\set{ag_1,ag_1b}$ and $\set{g_1b,ag_1b}$ considered in the LHS of \eqref{eq:edge}? This happens exactly if in steps $4$ and $5$ we will walk in  alternating colors ($A,B$ or $B,A$). Let $\EAB$ be the event that in step $4$ we choose an $A$-edge, i.e. $g_2=a_1g_1$ for some $a_1\in A$ and then in step $5$ we set $e'$ to be a $B$-edge, i.e. $e' = \set{a_1g_1,a_1g_1b_2}$ for some $b_2\in B$. Similarly let $\EBA$ be the   event that $g_2=g_1b_1$ and $e' = \set{g_1b_1,a_2g_1b_1}$. Clearly \[ \Pr[\EAB] = \Pr[\EBA] = \frac 1 4.
\]
Now,
\begin{equation}\label{eq:ab1}
    \Pr[\EAB\ve e'\in R] = \frac 1 4\cdot \Pr_{a_1,b_2}[\set{a_1g_1,a_1g_1b_2}\in R],
\end{equation}
and 
\begin{equation}\label{eq:ab2}
    \Pr[\EBA\ve e'\in R] = \frac 1 4\cdot \Pr_{a_2,b_1}[\set{g_1b_1,a_2g_1b_1}\in R].
\end{equation}
where the probability is taken over the randomness of the random process above conditioned on having entered step $3$. Since $\EAB$ and $\EBA$ are disjoint events,
\begin{align*}
     \Pr[e'\in R] &\geq 
     \Pr[ \EAB \ve e'\in R ]+\Pr[\EBA\ve e'\in R]\\
     &\geq \frac 1 4 \cdot (\Pr_{a,b}[\set{a g_1,ag_1b}\in R]+\Pr_{a,b}[\set{g_1b,ag_1b}\in R])\\
     &\geq \frac 1 4 \cdot
\Pr_{a,b}[\set{ag_1,ag_1b}\in R \hbox{ or }\set{g_1b,ag_1b}\in R]\\
&\geq \frac 1 4 \kappa_0\cdot r_i/|B| = \frac {r_i\kappa_0} {4|B|}
\end{align*}
where in the last inequality we have used \eqref{eq:edge}.
We conclude that if in step $3$ we choose $g_1=g$, then $\Pr[e'\in R]\geq\frac {r_1\kappa_0} {4|B|}$, whereas if in step $3$ we choose $g_1=ag$, then  $\Pr[e'\in R]\geq\frac {r_2\kappa_0} {4|B|}$. 

Altogether, recalling \eqref{eq:step2},
\[ \Pr[e'\in R] \ge \gamma\cdot \frac{r_0}{|B|} + (1-\gamma)\cdot \frac{\kappa_0}{ 4|B|}(r_1+r_2)/2.
\]
Plugging in $\gamma = \frac {\kappa_0} {8+\kappa_0}$, we get $1-\gamma ={8\gamma/\kappa_0}$, and recalling that $r_0+r_1+r_2\geq \delta_B|B|$,
\[\Pr[e'\in R] \ge \gamma (r_0+r_1+r_2)/|B| \geq \gamma \delta_B.
\]
We have seen that if $e=\set{g,ag}$ for some $a,g$ is in $R$, then $e'\in R$ with probability at least $\gamma \delta_B$. 
Symmetrically, if $e=\set{g,gb}$ for some $g,b$ is in $R$ then we would get that $e'\in R$ with probability at least $\gamma \delta_A$. Together this proves \eqref{eq:stay} and completes the proof of Lemma \ref{lemma:stay}.
\end{proof}
Recall from \eqref{eq:c} that $c = \frac{ \kappa_0}{8+\kappa_0} \cdot \min(\delta_A,\delta_B)$. By Lemma \ref{lemma:stay}, $\iprod{f,(\gamma\mpar + (1-\gamma)M)f} \geq c\cdot\iprod{f,f}$ so either 
\begin{equation}\label{eq:parlarge}
  \iprod{f,Mf}\geq c\iprod{f,f}  
\end{equation}
or
\begin{equation}\label{eq:Mlarge}
  \iprod{f,\mpar f}\geq c\iprod{f,f}.
\end{equation}

If \eqref{eq:parlarge} holds, then by Lemma \ref{lemma:AC-general}, applied with the operator $M$, whose vertex set $X(1)$ is endowed with the distribution $\D_1$, we get $\Pr_{\D_1}[R] \geq c - \lambda$ which means that $|R|\geq \frac{|G|}{2}\cdot\min(|A|,|B|)(c-\lambda)$.

Otherwise, assume that \eqref{eq:Mlarge} holds. 
By Lemma~\ref{lem:par} there exists some $\sigma\in \tilde A\cup \tilde B$ such that, $|R\cap X^\sigma(1)| \geq  |G| (c-\lambda)/2$.
 
This completes the proof of Proposition~\ref{prop:soundness} showing that if $\disagr(\W)>0$ then $\disagr(\W)> \frac{2(c-\lambda)}{|A|+|B|}$.
 \end{proof}

\section{A Concrete Construction}\label{sec:ins}
In the previous section we have described a code scheme: Given a left-right Cayley complex $Cay^2(A,G,B)$ together with two base codes $C_A\subseteq \bits^A$ and $C_B\subseteq\bits^B$, we get an error correcting code $C[G,A,B,C_A,C_B]$. 

In this section we prove our main theorem by showing how to find an infinite family of left-right Cayley complexes and base codes that yield locally testable codes.

\begin{theorem*}[Restatement of Theorem \ref{thm:main}]
For every $0<r<1$, there exist $\delta,\kappa>0$, $q\in \mathbb{N}$ and an explicit construction of an infinite family of error correcting codes $\set{C_n}_n$, such that for each $n$, $\Rate(C_n) \geq r$, $\dist(C_n) \geq \delta$ and $C_n$ is $\kappa$-locally testable with $q$ queries. 
\end{theorem*}
The proof of the theorem relies on the following two lemmas. 

\begin{lemma}[Good base code]\label{lem:GBC}
For every $0<r_0<1$, there exist $\delta_0, \kappa_0 > 0$ and $d_0, D_0 \in\mathbb{N}$, such that for every integer $D>D_0$ that is divisible by $d_0$, there exists a linear error correcting code $C_0\subseteq\bits^D$ with rate at least $r_0$, distance at least $\delta_0$, and such that the tensor code $C_0 \otimes C_0$ is $\kappa_0$-agreement testable.
\end{lemma}

\begin{lemma}[Good left-right Cayley complexes]\label{lem:Cay}
Let $d_0,D_0 \in \mathbb{N}$. 
Let $q$ be any odd prime power such that $q \geq \max \set{2d_0^2, D_0,17}$ and define $D = d_0 \cdot \lfloor \frac{q+1}{d_0} \rfloor$.
Then there exist an explicit construction of an infinite family of finite groups $G_i = PSL_2(q^i)$, with two symmetric generating subsets $A_i,B_i \subset G_i$, such that for every $i$, both $A_i$ and $B_i$ are of size $D$ hence divisible by $d_0$, $A_i$ and $B_i$ satisfy condition \eqref{eq:nc}, and the Cayley graphs $\mbox{Cay}(G_i,A_i)$ and $\mbox{Cay}(G_i,B_i)$ are $\lambda$-expanders where $ \lambda \leq 8 D^{-1/2}$.
\end{lemma}

We prove Lemma \ref{lem:GBC} in Subsection \ref{subsec:GBC} by showing that random LDPC codes are smooth. 
We prove Lemma \ref{lem:Cay} in Section \ref{sec:LRCC} using the known constructions of Ramanujan graphs by Lubotzky, Samuels and Vishne \cite{LSV2} and Morgenstern \cite{morgenstern1994existence}.

Let us now deduce Theorem \ref{thm:main} from Lemmas \ref{lem:GBC} and \ref{lem:Cay}.

\begin{proof}[Proof of Theorem \ref{thm:main}]
Fix $0<r<1$ and set $r_0 = \frac{{r+3}}{4}$ so that $r=4r_0-3$. 
By Lemma \ref{lem:GBC}, given $r_0$, there exist $\delta_0, \kappa_0 > 0$ and $d_0,D_0 \in \mathbb{N}$, such that for any $D>D_0$ divisible by $d_0$, there exists a code $C_0 \subset \bits^D$ with $\Rate(C_0) \geq r_0$, $\dist(C_0) \geq \delta_0$ and such that $C_0\otimes C_0$ is $\kappa_0$-agreement testable.

Define $q_0 = \max \lbrace 2D_0,\; 2d_0^2,\; 2^7\left(\frac{\kappa_0 + 8}{\kappa_0 \delta_0} \right)^2 \rbrace$.
For any $q \geq q_0$ odd prime power denote $D= d_0 \cdot \lfloor \frac{q+1}{d_0} \rfloor$.
Note that $q+1 \geq D \geq q + 1 - d_0 > q - \sqrt{q} > \frac{1}{2} q$. 
In particular $D > \frac{1}{2} q_0$, hence $8D^{-1/2} < 2^{7/2} q_0^{-1/2} \leq \frac{\kappa_0 \delta_0}{8 + \kappa_0}$, which also implies $8D^{-1/2} < \delta_0$. 

By Lemma \ref{lem:Cay} there exists an explicit construction of an infinite family of groups $G_i = PSL_2(q^i)$ together with generating sets $A_i,B_i$ such that for each $i\in\mathbb{N}$, $|A_i|=|B_i|=D $, conditions \eqref{eq:nc} holds, and both $Cay(G_i,A_i)$ and $Cay(G_i,B_i)$ are $\lambda = 8D^{-1/2}$ expanders.
In particular, from our choice of $D$, equation \eqref{eq:c} holds and $\lambda < \delta_0$.

By Lemma \ref{lem:GBC} there exists a code $C_0$ of length
$D$, with rate at least $r_0$, distance at least $\delta_0$ and such that the tensor code $C_0 \otimes C_0$ is  $\kappa_0$-agreement testable. Since $D$ is a constant we can, theoretically, enumerate over all possible codes in search of a good one.

Define our family of global codes to be $C_i = C[G_i,A_i,B_i,C_0,C_0]$, $i\in \mathbb{N}$, and by the above choices it has the following parameters:
\begin{itemize}
    \item Block-length $\frac{1}{4}|G_i|D^2$, where $|G_i| = \frac{1}{2}(q^{3i} - q^i)$. 
    \item Distance at least $\delta = \delta_0^2(\delta_0 - 4D^{-1/2}) > 0$, by Lemma \ref{lem:distance}, 
    \item Rate at least $r = 4r_0-3 > 0$, by Lemma \ref{lem:rate}, 
    \item It is $\kappa$-locally testable with $D^2$ queries, by Theorem \ref{thm:mainltc}, for 
    \begin{equation}
        \kappa = \min \left\lbrace \frac{1}{4+8D} \;,\; \frac{1}{4D} \left( \frac{\delta_0 \kappa_0}{8+\kappa_0} - 4D^{-1/2} \right) \right\rbrace .
    \end{equation}
\end{itemize}
Assuming that $\delta_0,\kappa_0$ are large with respect to $D$ we get distance $\approx \delta_0^3$ and detection probability $\kappa \approx \theta(\delta_0\kappa_0/D)$. 
\end{proof}

\subsection{Good Base Codes} \label{subsec:GBC}
In this section we prove Lemma \ref{lem:GBC} by relying on the notion of  smooth codes from \cite{DSW06}, which was consequently broadened to weakly-smooth codes in \cite{BVidweakly}. These works showed that the tensor product of a smooth code and any other code is robustly testable and therefore, by Lemma \ref{lem:robagr}, also agreement testable.

\begin{definition}[Smooth Code]
Let $c,d,n\in \mathbb{N},\alpha,\beta,\delta>0$. A $(c,d,n)$-LDPC code $C\subset
\bits^n$ is $(\alpha,\beta,\delta)$-smooth if for every $Y_0\subseteq Y$
with $|Y_0|\leq \alpha|Y|$ there is some $X_0\subseteq X$ with $|X_0|\leq
\beta|X|$ such that the code $C(\bar Y_0)|_{\bar X_0}$ has distance at least
$\delta$. 
Here the code $C(\bar{Y_0})|_{\bar X_0}$ is the code obtained by removing
the constraints in $Y_0$ and then removing the coordinates in $X_0$.
\end{definition}


\subsubsection{Random LDPC Codes}

We will next show that random LDPC codes satisfies w.h.p. the requirements of Lemma \ref{lem:GBC}. 

Random LDPC codes, see Definition \ref{def:ldpc}, were famously introduced by Gallager in his PhD thesis \cite{Gallager63}. 
A random $(c,d,n)$-code is given by selecting a random $(c,d)$-regular bipartite graph (called the factor graph of the code), which in turn is done by taking a random matching between the $nc$ ``half-edges'' on the left and the $md$ ``half-edges'' on the right, where we assume that $nc/d$ is an integer. 

Spielman described in his thesis \cite{Spielman96} the following expansion property,
\begin{definition}
A $(c,d)$-regular bipartite graph $([n],[m],E)$ is a $(\delta,\gamma)$-expander if every set of left vertices $A\subset [n]$ whose size is at most $\delta n$, has at least $c|A|(1-\gamma)$ neighbors.
\end{definition}

\begin{claim}[Claim 6.4 in \cite{Ben-SassonHR05}\footnote{A $(\delta,\gamma)$-expander here is called a $(c(1-\gamma),\delta)$-left-expander in \cite{Ben-SassonHR05}.}] \label{claim:cdexp-exist}
For any $d>c>2$, any $\gamma > \frac{1}{c}$ and any $n$ such that $nc/d \in \mathbb{N}$, then w.h.p. a random $(c,d)$-biregular graph $([n],[m],E)$ is a $(\delta,\gamma)$-expander for any $\delta$ satisfying
\begin{equation} \label{eq:BHR-delta}
\delta \leq \left(2 \cdot e^{c+1 - c\gamma} \cdot d^{c\gamma} \cdot (1 - \gamma)^{c\gamma} \right)^{-1/(c\gamma-1)}.
\end{equation}
Moreover, there exists at least one such graph for any $n \geq n_0$ divisible by $d$, where $n_0$ is the minimal integer satisfying
\begin{equation} \label{eq:BHR-n}
 e^{c+1 - c\gamma} \cdot d^{c\gamma} \cdot (1 - \gamma)^{c\gamma} \cdot n_0^{-1/9} + n_0 \cdot 2^{- n_0 ^{\min\{\frac{1}{3},\frac{c\gamma - 1}{2}\}}} < 1.
 \end{equation} 
\end{claim}

The proof is similar to Gallager's proof \cite{Gallager63} that a random LDPC code has constant distance with high probability.
Note that the upper bound on $n_0$ follows from showing that the last equation in the proof in \cite{Ben-SassonHR05}, which upper bound the probability that a random graph is not a $(\delta,\gamma)$-expander, is smaller then $1$.

\begin{claim} \label{lem:random-good}
A $(c,d,n)$-LDPC code whose factor graph is $(\delta,\gamma)$-expander graph, with $\gamma < \frac{1}{2}$, has rate at least $1 - \frac{c}{d}$ and distance at least $\delta$.
\end{claim}

\begin{proof}
The number of constraints of the code is $m=nc/d$, hence the dimension of the code is at least $n-m = n(1-\frac c d)$, so the rate is at least $1-\frac{c}{d}$. 
An LDPC code whose factor graph is a $(\delta,\gamma)$-expander with $\gamma < \frac{1}{2}$, has the unique neighbor expansion property \cite{Spielman96}, i.e. that for each subset $A \subset [n]$, $|A| \leq \delta n$, there exists  $u \in [m]$ with a unique neighbor in $A$, which implies that the distance of the code is at least $\delta$.
\end{proof}

In \cite{DSW06} and \cite{BVidweakly} it is shown that tensors of random LDPC codes are robustly testable,
\begin{theorem}[Robust testability of expander codes]\label{thm:tensors}
Let $C$ be a $(c,d,n)$-code whose factor graph is a $(c,d)$-regular  $(\delta,\gamma)$-expander. Let $C'$ be any linear code with distance $\delta'$. Then $C \otimes C'$ is $\tau$-robustly testable for
\begin{itemize}
    \item $\tau \ge \frac{\delta\delta' \cdot (\frac{1}{6} - \gamma)}{2d}$ when $\gamma<1/6$ \cite{DSW06}, and
    \item $\tau \ge \frac{\delta\delta' \cdot }{d^{\log_{0.5 + \gamma}0.05}}$ for all $\gamma<1/2$ \cite{BVidweakly}.
\end{itemize}
\end{theorem}

Finally, we can prove Lemma \ref{lem:GBC}, which we restate for convenience,

\begin{lemma*}[Restatement of Lemma \ref{lem:GBC}]
For all $0<r_0<1$, there exist $\delta_0, \kappa_0 > 0$ and $d_0, D_0 \in\mathbb{N}$, such that for every integer $D>D_0$ that is divisible by $d_0$, there exists a linear error correcting code $C_0\subseteq\bits^D$ with rate at least $r_0$, distance at least $\delta_0$, and such that the tensor code $C_0 \otimes C_0$ is $\kappa_0$-agreement testable.
\end{lemma*}

\begin{proof}
We fix $\gamma_0 = 0.15<1/6$ and set $c_0=7$ so that $\gamma_0 > 1/c_0$.  We choose  $d_0 = \lceil\frac 7 {1-r_0} \rceil$ such that $\frac{c_0}{d_0} \leq 1- r_0$. Claim \ref{claim:cdexp-exist} guarantees existence of $\delta_0>0$ and $D_0 = n_0$ such that for all $D>D_0$ divisible by $d_0$, a random $(c_0,d_0)$-regular bipartite graph with $D$ left vertices is a $(\delta_0,\gamma_0)$-expander with high probability. Moreover, we can take 

For such a bipartite graph, we take $C_0$ to be the corresponding $(c_0,d_0,D)$-LDPC code. By Lemma \ref{lem:random-good}, this code has rate at least $r_0$, distance at least $\delta_0$, and by taking $C' = C_0$ in Theorem \ref{thm:tensors}, we get that  $C_0 \otimes C_0$ is robustly testable with $\tau_0 = \frac{\delta_0^2 \cdot (\frac{1}{6} - \gamma_0)}{2d_0}$. 
By Claim \ref{claim:rtoag} these codes are $\kappa_0$-agreement testable for 
$\kappa_0 = \frac{\delta_0^3 \cdot (\frac{1}{6} - \gamma_0)}{4d_0}$.
\end{proof}

\begin{remark}\label{rem:params}
In Lemma \ref{lem:GBC}, the parameters $\delta_0$, $\kappa_0$ and $D_0$ depend on the parameter $\epsilon = 1- r_0$, as follows (ignoring absolute constants)
\[
\delta_0 = \left(2 \cdot e^{6.95} \cdot (\lceil 7\epsilon^{-1} \rceil)^{1.05} \cdot 0.85^{1.05} \right)^{-20} = \Omega(\epsilon^{21}),
\]
which follows directly from equation \eqref{eq:BHR-delta},
\[
\kappa_0 = \frac{\delta_0^3 \cdot 0.01666}{4d_0} = \Omega(\epsilon^{64}),
\]
and by bounding each of the two right terms in equation \eqref{eq:BHR-n} by $\frac{1}{2}$, we get
\[
D_0 \leq \max\{\left(2\cdot e^{6.95} \cdot (\lceil 7\epsilon^{-1} \rceil)^{1.05} \cdot 0.85^{1.05} \right)^9 , 2^{400} \} = O(\epsilon^{-9}).
\]
This shows that the parameters $\delta_0$, $\kappa_0$ and $D_0$ are polynomials in $\epsilon = 1 - r_0$.
\end{remark}

\section{Good Left-Right Cayley Complexes} \label{sec:LRCC}
In the previous section we showed how to construct good locally testable codes on left-right Cayley complexes provided the latter have sufficiently large spectral gap. 
To finish the proof of the main result of the paper, we neeed to show that such complexes indeed exist and to give an explicit construction.
Namely, in this section we prove Lemma \ref{lem:Cay}. 

More generally, we show that for every $\lambda > 0$, there exist $k_1, k_2 \in \mathbb{N}$ and an infinite family of finite groups $G_i$, with two symmetric subsets of generators $A_i, B_i$, such that for each $i$, $|A_i| = k_1$ and $|B_i| = k_2$, the two sets $A_i$ and $B_i$ satisfies \eqref{eq:nc}, and the second largest eigenvalues of the normalized adjacency matrices of $\mbox{Cay}(G_i,A_i)$ and $\mbox{Cay}(G_i,B_i)$, denoted $\lambda(\mbox{Cay}(G_i,A_i))$ and $\lambda(\mbox{Cay}(G_i,B_i))$, are bounded from above by $\lambda$.
Moreover, we can take $\lambda = \Theta(k_1^{-1/2}) = \Theta(k_2^{-1/2})$, making both Cayley graphs quasi-Ramanujan.

There are a number of ways in the literature to find Cayley graphs with small $\lambda(\mbox{Cay}(G,S))$. 
There are even various methods to give different sets of generators for the same group (see \cite{Lub10}, \cite{LSV2}). 
The difficulty is to ensure that condition \eqref{eq:nc} is satisfied. 
We will show two (actually three) ways to do so.
In all of our constructions, the elements in the sets $B_i$ will be of order $2$, while all the elements in $A_i$ will be of order greater then $2$. 
This ensures that \eqref{eq:nc} is automatically satisfied.

\subsection{The Morgenstern Generators, $q=2^\ell$}
In \cite{morgenstern1994existence}, Morgenstern presented for every prime power $q$, infinitely many groups $G_i = PGL_2(q^i)$ or $G_i = PSL_2(q^i)$ each with a symmetric set $B_i$ of $q+1$ generators such that $\mbox{Cay}(G_i,B_i)$ are Ramanujan, i.e., $\lambda(\mbox{Cay}(G_i,B_i))\leq\frac{2\sqrt{q}}{q+1}$.

The case of $q$ even, i.e., $q = 2^{\ell}$, is special in two ways.
First of all, here $PGL_2(q^i) = PSL_2(q^i)$, so this is always a simple group. 
But more importantly, in this case all the elements of $B_i$ are of order $2$ (see Remark \ref{rem:B_i} below).
Assume $q$ is even from now on.

Morgenstern constructed an explicit arithmetic lattice $\Gamma$ in the group $PSL_2(\mathbb{F}_q((t)))$ which is isomorphic to the free product $\langle b_0\rangle*\ldots*\langle b_q\rangle$, where $B=\{b_0,\ldots,b_q\}$ is a set of elements of order $2$ (see \cite[Section 5]{morgenstern1994existence}). 
The above mentioned Cayley graphs $\mbox{Cay}(G_i,B_i)$ are identified as quotients of this $\Gamma$ by normal congruence subgroups, where $B_i = \phi_i(B)$ is the image of $B$ under an epimorphism $\phi_i\,:\,\Gamma\rightarrow G_i$.
Note that by \cite{morgenstern1994existence} these Cayley graphs are all Ramanujan.

Let us now show how to get another symmetric set of generators $A_{i}$ for $G_{i}=PSL_{2}(q^{i})$ with $\lambda(\mbox{Cay}(G_{i},A_{i}))$ small, and such that $A_i$ and $B_i$ satisfy \eqref{eq:nc}.

Let $\Lambda$ be the index $2$ subgroup of $\Gamma$ - the kernel of the homomorphism $\phi\,:\,\Gamma\rightarrow C_2$ (= the cyclic group of order $2$) where $\phi$ sends each $b_j$ to the unique non-trivial element of $C_2$. 
One can see easily that $\Lambda$ is exactly the subgroup of all elements of $\Gamma$ of even length w.r.t. $B$. 
It is generated by the set $A=\{b_t b_s \;|\; b_t, b_s \in B,\; t\ne s\}$ which is of size $k_1 = q^2+q$. 
We claim

\begin{claim} \label{Claim:Morgenstern-A} 
(i) For $i > 1$, the image $A_i = \phi_i(A)$ of $A$ in $G_i$ generates $G_i = PSL_2(q^i)$.

(ii) $\lambda(\mbox{Cay}(G_i,A_i)) < \frac{3q-1}{q^2+q} < \frac{3\sqrt{k_1-1}}{k_1}$.

(iii) For $i > 4$, the images of the elements of $A$ in $G_i$ are distinct from one another, and each element in $A_i$ has order $>2$.
\end{claim}

\begin{proof}
(i) Since $\Lambda=\langle A\rangle$ is of index two in $\Gamma$ then $\langle A_{i}\rangle$ is of index at most two in $G_{i}$.
But $G_i = PSL_2(q^i)$ is simple, hence it has no index $2$ subgroup (a subgroup of index $2$ must be normal), which implies $\langle A_i \rangle = G_i$ .

(ii) Let $T_{B}$ and $T_{A}$ be the (non-normalized) adjacency matrices of $\mbox{Cay}(G_{i},B_{i})$ and $\mbox{Cay}(G_{i},A_{i})$, respectively. 
Note that $T_{B}^{2}=T_{A}+(q+1)I$.
Hence if $\mu$ is an eigenvalue of $T_{A}$, then $\mu=\lambda^{2}-(q+1)$ for some eigenvalue $\lambda$ of $T_{B}$. 
Since $\mbox{Cay}(G_{i},B_{i})$ is Ramanujan, $|\lambda|=q+1$ or $|\lambda|\leq2\sqrt{q}$. 
Therefore $\mu=q^{2}+q$ or $\mu\leq(2\sqrt{q})^{2}-(q+1)=3q-1$.

(iii) It suffices to show that each reduced word which is a product of length at most $4$ in $B$ is not in the kernel of $\phi_i$, which is equivalent to the girth of $\mbox{Cay}(G_i,B_i)$ being greater than $4$.
By \cite[Theorem 5.13 (3)]{morgenstern1994existence} the girth of $\mbox{Cay}(G_i,B_i)$ is at least $\frac{2}{3}\log_q|G_i|\geq i$, which completes the proof.
\end{proof}

Thus, given $\lambda>0$ by taking $q$ large enough so that $\frac{3\sqrt{q^{2}+q-1}}{q^{2}+q}<\lambda$, we get the desired $\lambda$-expanding left-right Cayley complexes with $k_{1}=q^{2}+q$ and $k_{2}=q+1$.

We can do slightly better. Note that $\Lambda$ above, being a normal subgroup of a free product of finite groups, with trivial intersection with each factor is a free group (see Section 34 in \cite{Kurosh}). 
In fact, by the Reidemeister-Schreier algorithm applied to the transversal set $\{1,b_{0}\}$ of $\Lambda$ in $\Gamma$ (or by inspection) one can see that $\Lambda$ is a free group on the $q$ generators $\{b_{0}b_{j}\;:\;j=1,\ldots,q\}$. 
As $(b_{0}b_{j})^{-1}=b_{j}b_{0}$ we deduce that $A'=\{b_{0}b_{j},\,b_{j}b_{0}\;:\;j=1,\ldots,q\}$ is a symmetric set of generators of $\Lambda$.

We can now look at the image $A'_i = \phi_i(A')$ under the epimorphism $\phi_i\,:\,\Gamma \rightarrow G_i$. 
Arguing similarly to the proof of Claim \ref{Claim:Morgenstern-A} (i), $A'_i$ generates $G_i$, and by the proof of (iii) above, the images are all different. 
Finally:

\begin{claim}\label{Claim:Morgenstern-A'} 
$\lambda(\mbox{Cay}(G_{i},A_{i}'))<\frac{3\sqrt{2q-1}}{2q}$.
\end{claim}

\begin{proof}
Let $V_{i}=\{f\,:\,G_{i}\rightarrow \mathbb{C}\}$ and for any element $s\in G_{i}$, define the $s$-adjecancy $T_{s}\,:\,V_{i}\rightarrow V_{i}$, $T_{s}f(g)=f(gs)$, and for any multiset $S$ of $G_{i}$, define the $S$-adjecancy operator $T_{S}\,:\,V_{i}\rightarrow V_{i}$, $T_{S}=\sum_{s\in S}T_{s}$.
Note that for any two multisets $S,S'$ of $G_{i}$, $T_{S\cup S'}=T_{S}+T_{S'}-T_{S\cap S'}$ and $T_{S}T_{S'}=T_{SS'}$, where $SS'=\{ss'\,:\,s\in S,s'\in S'\}$ counted with multiplicities. 
Therefore $T_{A_{i}'}=T_{b}T_{B_{i}}+T_{B_{i}}T_{b}-2I$, where $b=\phi_{i}(b_{0})$. 
Let $f\in V_{i}$ be such that $f\perp1_{G_{i}}$, i.e. $\sum_{g\in G_{i}}f(g)=0$. 
Note that for any $s\in G_{i}$, then $T_{s}f\perp1_{G_{i}}$ and $\|T_{s}f\|=\|f\|$. 
By \cite[Theorem 5.11]{morgenstern1994existence}, we have $\|T_{B_{i}}f\|\leq2\sqrt{q}\|f\|$ for any $f\perp1_{G_{i}}$.
Then 
\[
\|T_{A_{i}'}f\|\leq\|T_{b}T_{B_{i}}f\|+\|T_{B_{i}}T_{b}f\|+2\|f\|\leq\|T_{B_{i}}f\|+\|T_{B_{i}}(T_{b}f)\|+2\|f\|
\]
\[
\leq\left(2\sqrt{q}-1\right)\|f\|+\left(2\sqrt{q}-1\right)\|f\|+2\|f\|=4\sqrt{q}\|f\|\leq3\sqrt{2q-1}\|f\|
\]
which completes the proof.
\end{proof}

So this time we have a family of $\lambda$-expanders left-right Cayley complexes with $k_{1}=2q$ and $k_{2}=q+1$, for any $\lambda\geq\frac{3\sqrt{2q-1}}{2q}$.

\begin{remark} \label{rem:B_i} 
Everything said above is explicit. 
In fact the generator set $B_i$ of $PSL_2(q^i)$ are given explicitly in \cite[equation (21)]{morgenstern1994existence}.
Assume for simplicity that $i$ is even. 
Let $\textbf{i} \in \mathbb{F}_{q^i}$ be such that $\textbf{i} \not \in \mathbb{F}_q$ and $\epsilon = \textbf{i}^2 + \textbf{i} \in \mathbb{F}_q$.
Let $x\in\mathbb{F}_{q^i}$ be such that $1,x,\ldots,x^{e_i-1}$ form a basis for $\mathbb{F}_{q^i}$ over $\mathbb{F}_q$. 
Then the $q+1$ elements of $B_i$ are 
\begin{equation} \label{eq:B_i-Morgenstern}
\phi_{i}(b_j)=\left(\begin{array}{cc} 1 & \gamma_j + \delta_j\textbf{i}\\ x(\gamma_j + \delta_j + \delta_j\textbf{i}) & 1 \end{array}\right),\qquad j=0,\ldots,q,
\end{equation}
where $(\gamma_j,\delta_j)\in\mathbb{F}_q^2$ are the $q+1$ solutions in $\mathbb{F}_q$ for $\gamma^2 + \gamma\delta + \delta^2 \epsilon = 1$.
One indeed sees that each of the elements of $B_i$ is of order $2$.
\end{remark} 

We will pass now to a different construction, which will give us Cayley graphs of $G_{i}$ w.r.t. $A_{i}$ and $B_{i}$ of the same size:
$|A_{i}|=|B_{i}|=q+1$, and both are Ramanujan.

\subsection{The LSV Generators, $q$ odd}
In \cite{LSV2}, Lubotzky, Samuels and Vishne constructed Ramanujan complexes, based on an arithmetic lattice $\Gamma$, discovered
by Cartwright and Steger \cite{cartwright1998family}, which acts simply transitively on the Bruhat-Tits building of $PGL_{d}(\mathbb{F}_{q}((t)))$.
The special case $d=2$ gave some new Ramanujan graphs. 
These Ramanujan graphs were highlighted in \cite{kaufman2011edge}, as edge-transitive Ramanujan graphs which have been used there to construct symmetric LDPC codes. 

The arithmetic group $\Gamma$, acting simply transitively on the Bruhat-Tits tree of $PGL_{2}(\mathbb{F}_{q}((t)))$ ($q$ any odd prime power) is obtained there as a the group generated by the $q+1$ conjugates of a specific element $b$, conjugated by the non-split torus $T$ of order $q+1$ in $PGL_{2}(\mathbb{F}_{q})$. 
This is a symmetric set of generators $A$ for $\Gamma$ which generates a free group on $\frac{q+1}{2}$ generators.
We will present below a different choice for $b$, this time $b'$ - an element of order $2$, whose conjugation under $T$ forms a symmetric
set $B$ of size $q+1$ and generate a group $\Gamma'$ which also acts simply transitively on the Bruhat-Tits tree. 
Moreover, $\Gamma$ and $\Gamma'$ are both finite index subgroups of an arithmetic group $G(R)$ - to be defined below.

In \cite{LSV2} (see also \cite{kaufman2011edge}) it was shown that $G(R)$ has infinitely many finite congruence quotients $G_{i}$, under the maps $\phi_{i}\,:\,G(R)\rightarrow G_{i}$, where $G_i = PGL_2(\F_{q^i})$ or $G_i = PSL_2(\F_{q^i})$, for which $\mbox{Cay}(G_{i},\phi_{i}(A))$ are Ramanujan $(q+1)$-regular graphs. 
We will observe below that the same holds for $\mbox{Cay}(G_{i},\phi_{i}(B))$.
For $i$ large enough (see Claim \ref{Claim:LSV-generators-order}) the elements of $\phi_{i}(A)$ are of order $>2$ while $\phi_{i}(B)$ contains only elements of order $2$. 
Hence we will get two-sided Cayley square complexes with $k_{1}=k_{2}=q+1$ and $\lambda\leq\frac{2\sqrt{q}}{q+1}$.
By choosing $q$ large enough, they will be $\lambda$-expanders for arbitrarly small $\lambda>0$. 

Now, in more details: Let $0\ne\epsilon\in \mathbb{F}_{q}$ be a non-square element, let $R=\mathbb{F}_{q}[y,\frac{1}{y},\frac{1}{1+y}]$ be the subring of $\mathbb{F}_{q}(y)$, generated by $y$, $\frac{1}{y}$ and $\frac{1}{1+y}$, and let $A(R)$ be the quaternion $R$-algebra,
\begin{equation}
A(R)=R+R\alpha+Rz+R\alpha z\qquad:\qquad\alpha^{2}=\epsilon,\quad z^{2}=1+y,\quad z\alpha=-\alpha z.\label{eq:A(R)}
\end{equation}
 
\begin{remark}
We note that our choice of basis for $A(R)$, $\{1,\alpha,z,\alpha z\}$, is based on \cite{kaufman2011edge}, while \cite{LSV2} used a different basis for $A(R)$, $\{\xi,\xi^{q},\xi z,\xi^{q}z\}$, where $\{\xi,\xi^{q}\}$ forms an $\mathbb{F}_{q}$-basis for $\mathbb{F}_{q^{2}}=\mathbb{F}_{q}[\alpha]$.
The change of bases does not affect any of the following constructions. 
\end{remark}

For any ring $D$, denote by $D^{*}$ its group of units. 
Note that an element of $r(y)\in R$ belongs to $R^{*}$ if and only if it is of the form $r(y)=cy^{n}(1+y)^{m}$ , $c\in\mathbb{F}_{q}^{*}$, $n,m\in\mathbb{Z}$, and that an element $a=a_{1}+a_{2}\alpha+a_{3}z+a_{4}\alpha z\in A(R)$ belongs to $A(R)^{*}$ if and only if its norm $N(a):=a_{1}^{2}-\epsilon a_{2}^{2}-(1+y)a_{3}^{2}-\epsilon(1+y)a_{4}^{2}\in R$ belongs to $R^{*}$. 
Note also that $R$ is the center of $A(R)$ and $R^{*}$ is the center of $A(R)^{*}$. 
Then the principal arithmetic group $G(R)$ is defined to be
\[
G(R)=A(R)^{*}/R^{*}=\left\{ a\in A(R)\,:\,N(a)\in R^{*}\right\} /R^{*}.
\]
The Cartwright--Steger arithmetic lattice $\Gamma$, and the second arithmetic lattice $\Gamma'$, are defined to be the subgroups of $G(R)$, generated by the symmetric sets of size $q+1$, $A$ and $B$, which are the sets of $T$ conjugates of the elemenets, $b$ and $b'$, respectively, where $T=\mathbb{F}_{q}[\alpha]^{*}/\mathbb{F}_{q}^{*}\leq G(R)$ is a non-split torus of order $q+1$, $b=\left(1-\frac{1}{1+y}z\right)R^{*}\in G(R)$ and $b'=\alpha b=\left(\alpha-\frac{1}{1+y}\alpha z\right)R^{*}\in G(R)$, namely, 
\[
\Gamma=\langle A\rangle\leq G(R),\quad A=\left\{ tbt^{-1}\,:\,t\in T\right\} ,\quad \Gamma'=\langle B\rangle\leq G(R),\quad B=\left\{ tb't^{-1}\,:\,t\in T\right\} ,
\]

\[
T=\mathbb{F}_{q}[\alpha]^{*}/\mathbb{F}_{q}^{*},\qquad b=\left(1-\frac{1}{1+y}z\right)R^{*},\qquad b'=\alpha b=\left(\alpha-\frac{1}{1+y}\alpha z\right)R^{*}.
\]
Note that $b$ and $b'$ belongs to $G(R)$, since $N\left(1-\frac{1}{1+y}z\right)=1-(1+y)\frac{1}{(1+y)^{2}}=\frac{y}{1+y}\in R^{*}$ and $N\left(\alpha-\frac{1}{1+y}\alpha z\right)=N(\alpha)\cdot N\left(1-\frac{1}{1+y}z\right)=-\epsilon\cdot\frac{y}{1+y}\in R^{*}$.

\begin{claim} \label{Claim:LSV-generators-order} 
(i) Every element of $A$ is of infinite order, while every element of $B$ is of order $2$.

(ii) For $i > 2$, every element of $A_{i}=\phi_{i}(A)$ is of order $>2$, while every element of $B_{i}=\phi_{i}(B)$ is of order $2$.
\end{claim}

\begin{proof}
(i) The claim about the elements of $A$ follows from \cite[Corollary 5.4]{LSV2}.
For the claim about the elements of $B$, since they are all conjugate of one another, it suffice to show $b'^{2}=1$, or equivalently, $\left(\alpha-\frac{1}{1+y}\alpha z\right)^{2}\in R^{*}$.
This follows from the following computations, 
\[
\left(\alpha-\frac{1}{1+y}\alpha z\right)^{2}=\alpha^{2}-\frac{1}{1+y}\alpha\alpha z-\frac{1}{1+y}\alpha z\alpha+\frac{1}{(1+y)^{2}}\alpha z\alpha z=*,
\]
and by Equation \ref{eq:A(R)}, as $\alpha z=-z\alpha$, $\alpha^{2}=\epsilon$ and $z^{2}=1+y$, we get
\[
*=\alpha^{2}-\frac{1}{(1+y)^{2}}\alpha^{2}z^{2}=\epsilon-\frac{\epsilon}{1+y}=\epsilon\frac{y}{1+y}\in R^{*}.
\]

(ii) This follows from an injectivity radius argument for congruence subgroups, see for instance \cite{lubotzky2007moore}.
\end{proof}

Let $\mathcal{B}$ be the Bruhat-Tits tree of $PGL_{2}(\mathbb{F}_{q}((t)))$, which is a $(q+1)$-regular infinite tree. 
By \cite[Section 3]{LSV2}, $\Gamma$, $\Gamma'$ and $G(R)$ are subgroups of $PGL_{2}(\mathbb{F}_{q}((t)))$, hence acts on $\mathcal{B}$. 
In the notation of \cite{LSV2}, let $v_{0}=[L_{0}]$ be the fundamental vertex in $\mathcal{B}$, and let $\Omega$ be the set of its neighbors.

\begin{claim} \label{Claim:LSV-generators-Cayley-Ramanujan} 
\begin{enumerate}[(i)]
\item For each set $X=A$ or $X=B$, the map $g\leftrightarrow g.v_{0}$ is a bijection between $X$ and $\Omega$. 
\item The subgroups, $\Gamma$ and $\Gamma'$, acts simply transitively on the Bruhat-Tits tree.
\item Both subgroups, $\Gamma$ and $\Gamma'$, are normal in $G(R)$ and of index $2(q+1)$. 
\item If $\phi:G(R)\to PSL_{2}(q^{e})$ is an epimorphism, then both subsets, $\phi(A)$ and $\phi(B)$, are symmetric set of generators for $PSL_{2}(q^{e})$.
\item If $\phi:G(R)\to PSL_{2}(q^{e})$ is an epimorphism whose kernel is a congruence subgroup $G(R,\phi)$ of $G(R)$, then both Cayley graphs, $\mbox{Cay}(PSL_{2}(q^{e}),\phi(A))$ and $\mbox{Cay}(PSL_{2}(q^{e}),\phi(B))$, are Ramanujan $(q+1)$-regular graphs.
\end{enumerate}
\end{claim}

\begin{proof}
\begin{enumerate}[(i)]
    \item The claim for $A$ is \cite[Proposition 4.3]{LSV2}.
The claim for $B$ follows from the claim for $A$ and the identity $tb't^{-1}.v_{0}=t\alpha bt^{-1}.v_{0}=(t\alpha)b(t\alpha)^{-1}(t\alpha t^{-1}).v_{0}$.
Now, $\alpha\in T$ and $T$ fixes $v_{0}$, so $\{tb't^{-1}.v_{0}|t\in T\}=\{tbt^{-1}.v_{0}|t\in T\}$.
\item The transitivity claim for $\Gamma$ is \cite[Proposition 4.5]{LSV2}, which relies solely on the validity of claim (i) for the generating set $A$ of $\Gamma$, hence the same proof works also for $\Gamma'$.
Moreover, the same proof can actually show that for any $n\in\mathbb{N}$, for any vertex $v$ of distance $n$ from $v_{0}$, there exists a reduced word $g=s_{1}\cdots s_{n}\in\Gamma$ (resp. $\Gamma'$), $s_{1},\ldots,s_{n}\in A$ (resp. $B$), such that $g.v_{0}=v$. 
This proves that the action is also simply transitive since the number of vertices of distance $n$ is equal the number of reduced words of length $n$, for any $n\in\mathbb{N}$.
\item The claim for $\Gamma$ follows from \cite[Propositions 4.9 and 3.5]{LSV2}, and the same proof also works for $\Gamma'$. 
The fact that the index is $2(q+1)$ follows also from the fact that $\Gamma'$ acts simply transitively on the Bruhat-Tits tree by (ii).
Hence the index of $\Gamma'$ in $G(R)$ is equal to the order of the stabilizer of $v_{0}$ in $G(R)$, which by \cite[Proposition 3.5]{LSV2}, is of size $2(q+1)$.
\item By (iii) both images, $\phi(\Gamma)$ and $\phi(\Gamma')$, are normal subgroups of index $\leq2(q+1)$ in $PSL_{2}(q^{e})$, and since $PSL_{2}(q^{e})$ is a simple group of size $\geq\frac{1}{2}(q+1)q(q-1)>2(q+1)$, we get that $\phi(\Gamma)=PSL_{2}(q^{e})=\phi(\Gamma')$.
\item The claim for $\mbox{Cay}(PSL_{2}(q^{e}),\phi(A))$ is \cite[Theorem 7.1]{LSV2}, and the same proof holds also for $\mbox{Cay}(PSL_{2}(q^{e}),\phi(B))$.
Another way to see this is to observe that both graphs are isomorphic to $G(R,\phi)\backslash \mathcal{B}$ and in particular they are isomorphic, so if one is Ramanujan so is the other.
\end{enumerate}
\end{proof}

\subsection{Proof of Lemma \ref{lem:Cay} and Degree Reduction}

First we use the LSV generators constructed in the previous subsection to prove the following Lemma.

\begin{claim}\label{claim:LSV-graphs}
For any odd prime power $q$ there exist an explicit construction of an infinite family of finite groups $G_i = PSL_2(q^i)$, with two symmetric generating subsets $A_i,B_i$ of $G_i$, such that for each $i$, $|A_i|=|B_i|=q+1$, condition \eqref{eq:nc} holds for $A_i$ and $B_i$, and the Cayley graphs $\mbox{Cay}(G_i,A_i)$ and $\mbox{Cay}(G_i,B_i)$ are Ramanujan, in particular they are $\lambda$-expanders with $ \lambda \leq 2 (q+1)^{-1/2}$.
\end{claim} 

\begin{proof}
From Claim \ref{Claim:LSV-generators-Cayley-Ramanujan} we get that for any $i$, there exists two symmetric generating subsets $A_i$ and $B_i$ of the finite group $G_i = PSL_2(q^i)$, both sets are of size $q+1$, and the Cayley graphs $\mbox{Cay}(G_i,A_i)$ and $\mbox{Cay}(G_i,B_i)$ are both Ramanujan.
By Claim \ref{Claim:LSV-generators-order}, for any $i > 2$, the two sets $A_i$ and $B_i$ satisfy condition \eqref{eq:nc}.
\end{proof}

Next we prove the following degree reduction trick, which allows us to start with a $\lambda$-expander Cayley graph, and to remove a few elements from the generating set with only negligible effect on $\lambda$.

\begin{claim} \label{claim:degree-reduction}
(i) Let $G$ be a finite group, let $S' \subset S$ be two symmetric subsets of $G$, and denote $\lambda = \lambda(\mbox{Cay}(G,S))$ and $\lambda' = \lambda(\mbox{Cay}(G,S'))$ the normalized second largest eigenvalues of the corresponding Cayley graphs. 
Then 
\[
\lambda' \leq \lambda + 2 \frac{|S \setminus S'|}{|S'|}.
\]
(ii) In particular, if $|S\setminus S'| \leq 2 |S|^{1/2}$ and  $ 2 |S|^{-1/2} \leq \lambda \leq \frac{1}{3}$, then
\[
\lambda' \leq 4 \lambda.
\]
\end{claim}

\begin{proof}
(i) Let $M = M_S$ and $M' = M_{S'}$ be the adjacency matrices of $\mbox{Cay}(G,S)$ and $\mbox{Cay}(G,S')$, respectively.
Since $\mbox{Cay}(G,S)$ is $|S|$-regular (resp. $\mbox{Cay}(G,S')$ is $|S'|$-regular), the largest eigenvalue of $M$ is $|S|$ (resp. $M'$ is $|S'|$), with corresponding eigenvector the constant function $1_G$.
By the Courant-Fischer Formula we get that 
\[
\lambda \cdot |S| = \max_{0 \ne v \perp 1_G} \frac{v^tMv}{v^tv} \qquad \mbox{and} \qquad \lambda' \cdot |S'| = \max_{0 \ne v \perp 1_G} \frac{v^tM'v}{v^tv}.
\]
Now the matrix $M-M'$ can be considered as the adjacency matrix of $\mbox{Cay}(G,S\setminus S')$, which by the Perron-Frobenius Theorem, all of its eigenvalues are bounded in absolute value by $|S \setminus S'|$, and by the Courant-Fischer Formula $|S \setminus S'| = \max_{0 \ne v} \frac{v^t(M' - M)v}{v^tv} $. 
Therefore we get that
\[
\lambda' \cdot |S'| = \max_{0 \ne v \perp 1_G} \frac{v^tM'v}{v^tv} 
\leq \max_{0 \ne v \perp 1_G} \frac{v^tMv}{v^tv} + \max_{0 \ne v \perp 1_G} \frac{v^t(M' - M)v}{v^tv}
\leq \lambda \cdot |S| + |S \setminus S'|,
\]
and after dividing by $|S'|$ we get claim (i),
\[
\lambda' \leq \lambda \cdot \frac{|S|}{|S'|} + \frac{|S\setminus S'|}{|S'|}
\leq \lambda + (1 + \lambda) \frac{|S\setminus S'|}{|S'|}
\leq \lambda + 2 \frac{|S\setminus S'|}{|S'|}.
\]

(ii) First note that
\[
|S \setminus S'| \leq  2 |S|^{1/2} \leq \lambda |S| 
\leq \lambda |S'| + \lambda |S \setminus S'|
\quad \Rightarrow \quad
\frac{|S \setminus S'|}{|S'|} \leq \frac{\lambda}{1 - \lambda} \leq \frac{3 \lambda}{2}.
\]
Combined with claim (i) we get claim (ii).
\end{proof}

Finally we combine the above two Claims to prove Lemma \ref{lem:Cay}.

\begin{lemma*}[Restatement of Lemma \ref{lem:Cay}]
Let $d_0,D_0 \in \mathbb{N}$. 
Let $q$ be any odd prime power such that $q \geq \max \set{2d_0^2, D_0,17}$ and define $D = d_0 \cdot \lfloor \frac{q+1}{d_0} \rfloor$.
Then there exist an explicit construction of an infinite family of finite groups $G_i = PSL_2(q^i)$, with two symmetric generating subsets $A_i,B_i \subset G_i$, such that for each $i$, both $A_i$ and $B_i$ are of size $D$ hence divisible by $d_0$, $A_i$ and $B_i$ satisfy condition \eqref{eq:nc}, and the Cayley graphs $\mbox{Cay}(G_i,A_i)$ and $\mbox{Cay}(G_i,B_i)$ are $\lambda$-expanders where $ \lambda \leq 8 D^{-1/2}$.
\end{lemma*}

\begin{proof}[Proof of Lemma \ref{lem:Cay}]
First note that $D$ is by definition the largest integer $\leq q+1$ which is divisible by $d_0$, and that $q+1-D \leq d_0 \leq \frac{1}{2} \sqrt{D}$.

By Claim \ref{claim:LSV-graphs}, for each $i$, there exist  $\tilde A_i,\tilde B_i$ two symmetric generating subsets of $G_i = PSL_2(q^i)$, such that $\tilde A_i,\tilde B_i$ are both of size $q+1$, they satisfy \eqref{eq:nc} and such that the corresponding Cayley graphs are Ramanujan, i.e. $\lambda$-expanders for $\lambda\leq \frac{2\sqrt{q}}{q+1}\leq 2(q+1)^{-1/2}$.
Since $q\geq 17$, then $\lambda \leq \frac{1}{3}$.

Let $A_i \subset \tilde A_i$ and $B_i \subset \tilde B_i$ be any two symmetric subsets of size $D$.
Since $\tilde A_i$ and $\tilde B_i$ satisfy \eqref{eq:nc}, any subsets of them must also satisfy \eqref{eq:nc}. 

By Claim \ref{claim:degree-reduction}, we get that for $G= G_i$,  $S= \tilde A_i$ or $\tilde B_i$, and $S' = A_i$ or $B_i$, respectively, we get that
\[
\lambda(\mbox{Cay}(G,S')) \leq 4 \lambda(\mbox{Cay}(G,S)) \leq 8 D^{-1/2},
\]
which completes the proof of the Lemma.
\end{proof}

\section*{Acknowledgements}
We wish to thank Prahladh Harsha and Avi Wigderson for many interesting discussions along the way of this project. We also wish to thank Tali Kaufman for her influential role in connecting LTCs and high dimensional expansion.

This work was presented by the first author on October 6, 2021 
at the Simon's Institute for the Theory of Computing \cite{breakthrough-talk} as part of the lecture series on breakthroughs in computer science, and at the Institute for Advanced Study in Princeton on October 25-26, 2021 \cite{IAStalk}. It was also presented by the fourth author on October 27, 2021 at the Simon's HDX21 workshop \cite{SimonsHDX-talk}. The authors are very grateful to these institutions and for the remarks of the audience which improved the exposition of the paper. 

Irit Dinur acknowledges support by ERC grant 772839 and ISF grant 2073/21. Shai Evra is grateful to the Azrieli Foundation for the award of an Azrieli Fellowship.
Alexander Lubotzky's research is supported by a grant from the Institute for Advanced Study at Princeton and by the European Research Council (ERC) under the European Union’s Horizon 2020 research and innovation programme (grant agreement No 882751).
Shahar Mozes acknowledges support by ISF-Moked grant 2019/19.


\newcommand{\etalchar}[1]{$^{#1}$}

\appendix
\section{Robust Testability and Agreement Testability }\label{app:robagr}
In this section we show the equivalence between the two notions, proving Lemma \ref{lem:robagr}

\begin{claim}[Robust testability implies agreement testability]\label{claim:rtoag} 
Assume  $\delta_i=\dist(C_i)$ for $i=1,2$.
If $C_1\otimes C_2$ is $\tau$-robustly testable then  $C_1\otimes C_2$ is $\kappa$-agreement testable, for  $\kappa = \frac {2\tau\delta_1\delta_2}{\delta_2 + \delta_1(1+2\tau)}$.
\end{claim}

\begin{proof}
Suppose  $w_1 \in C_1\otimes \bits^{n_2}$, and $w_2\in \bits^{n_1}\otimes C_2$. Let $f=w_1$, so $\dcol(f)=0$, and observe that since $w_2(i,\cdot)\in C_2$ for each $j$,
\[ \drow(f) = \E_{i\in [n_1]} \dist(f(i,\cdot),C_2)
\leq \E_{i\in [n_1]} \dist(f(i,\cdot),w_2(i,\cdot)) = \dist(w_1,w_2).
\]
By the robust testability of $C_1\otimes C_2$ there is some $w\in C_1\otimes C_2$ such that 
\[ \dist(w,w_1)=\dist(w,f) \leq \frac 1 \tau\cdot \frac{\drow(f)+\dcol(f)}2 \leq \frac 1 {2\tau} \cdot (\dist(w_1,w_2)+0).
\]
By the triangle inequality $\dist(w,w_2) \leq \dist(w,w_1) + \dist(w_1,w_2)\leq (1+\frac 1 {2\tau})\dist(w_1,w_2)$. 

Next, observe that $\Pr_j[w(\cdot,j)\neq w_1(\cdot,j)]\cdot \delta_1 \leq \dist(w,w_1)$, and
similarly
$\Pr_i[w(i,\cdot)\neq w_2(i,\cdot)]\cdot \delta_2 \leq \dist(w,w_2)$.
Altogether, 
\begin{align*}
     \Pr_j[w(\cdot,j)\neq w_1(\cdot,j)]+
     \Pr_i[w(i,\cdot)\neq w_2(i,\cdot)] &\leq \frac 1{\delta_1} \dist(w,w_1) + \frac 1 {\delta_2}\dist(w,w_2) \\
    &\leq (\frac 1 {2\tau\delta_1} + \frac {1+1/(2\tau)} {\delta_2} )\cdot  \dist(w_1,w_2) 
\end{align*}
proving the claim with $\kappa^{-1} = \frac 1 {2\tau\delta_1} + \frac {1+1/(2\tau)} {\delta_2}$, or $\kappa = \frac {2\tau\delta_1\delta_2}{\delta_2 + \delta_1(1+2\tau)}$.
\end{proof}

Note that in case $\delta_1=\delta_2=\delta$ the statement simplifies slightly to $\kappa = \frac {\tau\delta}{\tau+1}$.
The other direction, that we do not need here, is even simpler,

\begin{claim}[Agreement testability implies robust testability]\label{claim:agtor} 
If $C_1\otimes C_2$ is $\kappa$-agreement testable, then   $C_1\otimes C_2$ is $\tau$-robustly testable for $\tau = \frac\kappa{2(\kappa+1)} $.
\end{claim}

\begin{proof}
Assume $C_1\otimes C_2$ is $\kappa$-agreement testable. Let $w\in \bits^{n_1\times n_2}$ satisfy $\delta(w)=\frac{\delta^{col}(w)+\delta^{row}(w)}{2}=\delta$. Let $w_1 \in C_1\otimes \bits^{n_2}$ be such that $\delta^{col}(w) = \dist(w,w_1)$. Let $w_2\in \bits^{n_1}\otimes C_2$ be such that $\delta^{row}(w) = \dist(w,w_2)$. 
By the triangle inequality,
\[ \dist(w_1,w_2) \leq \dist(w_1,w)+\dist(w,w_2) = \delta^{col}(w)+\delta^{row}(w)= 2\delta(w).
\]
By the $\kappa$-agreement testability there is some $w'\in C_1\otimes C_2$ such that
\[ 
\kappa \cdot (\Pr_i[w_1(i,\cdot)\neq w'(i,\cdot)] + \Pr_j[w_2(\cdot,j)\neq w'(\cdot,j)])
\leq \Pr_{i,j}[w_1(i,j)\neq w_2(i,j)]) = \dist(w_1,w_2) \leq 2\delta(w).
\]
But clearly 
\begin{equation}\label{eq:ww}
\dist(w_1,w')+\dist(w',w_2) \leq  \Pr_i[w_1(i,\cdot)\neq w'(i,\cdot)] + \Pr_j[w_2(\cdot,j)\neq w'(\cdot,j)] 
\end{equation}
so again by the triangle inequality,
\begin{align*}
\dist(w,w') &\leq \frac 1 2 (\dist(w,w_1) + \dist(w_1,w') + \dist(w,w_2) + \dist(w_2,w')) \\
&=  \frac 1 2 (\dist(w,w_1) + \dist(w,w_2) + \dist(w_1,w') + \dist(w_2,w')) \\
&\leq \delta(w) + \kappa^{-1}\cdot \delta(w) = \frac{\kappa+1}\kappa \cdot \delta(w).
\end{align*}
\end{proof}

\end{document}